\documentclass[pdflatex,sn-mathphys]{sn-jnl}

\usepackage{algorithm}
\usepackage{graphicx}
\usepackage[font=footnotesize,labelfont=bf]{caption}
\usepackage{xcolor}
\definecolor{mygray}{gray}{0.95}
\usepackage{subcaption}

\usepackage{mathtools}
\usepackage{array}
\usepackage{multirow}
\usepackage{bigstrut}

\usepackage{bm}
\usepackage{commath}
\usepackage{amsmath}
\usepackage{empheq}

\theoremstyle{thmstyleone}
\newtheorem{theorem}{Theorem}
\newtheorem{lemma}[theorem]{Lemma}
\newtheorem{corollary}{Corollary}[theorem]
\newtheorem{assumption}{Assumption}
\newtheorem{conjecture}[theorem]{Conjecture}

\newtheorem{observation}[theorem]{Observation}

\theoremstyle{thmstyletwo}

\newtheorem{remark}{Remark}

\theoremstyle{thmstylethree}

\definecolor{Red}{RGB}{255, 99, 71}      
\definecolor{Orange}{RGB}{255, 165, 0}     
\definecolor{Blue}{RGB}{30, 144, 255}   
\definecolor{Green}{RGB}{34, 139, 34}

\newcommand{\ket}[1]{ \vert #1 \rangle }
\newcommand{\bra}[1]{ \langle #1 \vert }

\newcommand{\varphiequal}{p \,w_{s}\,c _{vp}}
\newcommand{\nnm}{\nonumber\\}

\newcommand\gammabeta{$(\vec{\gamma},\vec{\beta})\;$}
\newcommand{\costptheta}{F_p(\vec{\theta})}
\newcommand{\FpThetaG}{F_p(\vec{\theta};G)}
\newcommand{\FpThetaGm}{F_p(\vec{\theta};G_m)}

\newcommand{\vtheta}{\vec{\theta}}

\newcommand{\shots}{n_p}

\newcommand{\rvpsym}{c_{vp}}

\newcommand{\ER}{Erdős–Rényi }
\newcommand{\ERc}{Erdős–Rényi, }
\newcommand{\BS}{Benjamini-Schramm }
\newcommand{\hgt}{\hat g^t}
\newcommand{\gttheta}{\hat{g}(\vec{\theta}^{\,t})}

\newcommand{\costexpect}{\bigl\vert\langle C_p\rangle\bigr\vert}

\newcommand{\costsquare}{\langle C_p\rangle ^ 2}

\usepackage[normalem]{ulem}
\usepackage{color,soul}

\usepackage{mathtools}
 
\begin{document}

\title[Measurements Number Scaling in QAOA: A Statistical Analysis]{Measurements Number Scaling in the Quantum Approximate Optimization Algorithm for MaxCut: A Statistical Analysis}

\author[1]{\fnm{Inbar} \sur{Chefer}}

\author*[1]{\fnm{Uri} \sur{Shaham}}\email{uri.shaham@biu.ac.il}

\author*[2]{\fnm{Adi} \sur{Makmal}}\email{adi.makmal@biu.ac.il}

\affil[1]{\orgdiv{Department of Computer Science}, \orgname{Bar-Ilan University}, \orgaddress{\city{52900 Ramat-Gan}, \country{Israel}}}

\affil[2]{\orgdiv{The Engineering Faculty}, \orgname{Bar-Ilan University}, \orgaddress{\city{52900 Ramat-Gan}, \country{Israel}}}

\abstract{
We provide a statistical analysis of the measurement (shot) requirements of the quantum approximate optimization algorithm (QAOA) for the MaxCut problem.
We derive sufficient conditions on the number of shots per cost operator evaluation to: (a) estimate the expected cost to within a relative error $\delta$ and a confidence $1-\epsilon$, and (b)  ensure SGD-based parameter optimization converges to a target relative suboptimality level  with high probability. 
In addition, we provide an explicit bound on the number of SGD iterations required to reach the target accuracy.
Our analysis reveals an unexpected scaling phenomenon: for specific graph classes, which we formally characterize, the total shot budget needed to achieve a fixed relative-performance metric \textit{decreases} as the instance size grows.
This result complements earlier 
cost function concentration arguments regarding parameter optimization redundancy, thereby highlighting the potential for high-performance, low-overhead QAOA implementations for large-scale MaxCut instances. 
To assist practitioners, we translate our analytical findings into practical rules of thumb for shot-budget allocation and validate these results with numerical simulations, offering new insights into the interplay between graph size, structural complexity, and resource requirements in QAOA.}


\maketitle
\section{Introduction}\label{sec:intro}
The quantum approximate optimization algorithm (QAOA) has emerged as a leading quantum approach for solving combinatorial optimization problems, which are fundamental to many computational tasks \cite{farhi2014quantum, Kashapogu_2024wed}. QAOA belongs to the class of variational quantum algorithms (VQAs) \cite{cerezo2021variational}, which utilize a tunable quantum circuit whose parameters are iteratively optimized using classical techniques such as stochastic gradient descent (SGD) to approach an optimal solution. As its name suggests, QAOA is an approximation algorithm, providing analytical bounds on its performance, in terms of the approximation ratio.
Moreover, at shallow circuit depths, QAOA is particularly well-suited for current and near-term noisy quantum hardware, making it a potential candidate for demonstrating quantum advantage over classical computation even before the development  of fully fault-tolerant quantum technology.
The time-to-solution (TTS) of QAOA scales linearly with the number of circuit layers ($p$),  iterations ($I$), and measurements (shots) per iteration ($n_I$) required to achieve statistical accuracy. 
Much analysis has been conducted on these quantities \cite{brandao2018fixed,zhou2020quantum,wauters2020polynomial}. For instance, it is known that the required number of layers $p$ depends on both the problem structure and the desired approximation ratio \cite{farhi2014quantum, wang_quantum_2018}. In addition, numerical studies indicate that convergence may be slow due to trainability barriers \cite{rajakumar2024trainability} and  barren plateaus (see \cite{larocca2025barren} and references therein) but that the number of iterations may be reduced using warm start or by means of machine learning \cite{zhou2020quantum,Sack_2021, amosy2024iteration}. 
In so called \emph{tree-QAOA} approaches, one leverages the fact that for small $p$ the causal cone of a local cost term is often tree-like: it is exactly so on high-girth regular graphs, and asymptotically so on regular and sparse \ER graph families.
This permits efficient contraction or dynamic-programming-style evaluation of local observables on the induced tree subgraph, yielding instance-independent or weakly instance-dependent angle choices, and thus the outer optimization loop can be completely avoided \cite{brandao2018fixed,streif2019treeqaoa,wurtz2021fixedangle,wybo2025missing}.
A recent result for high-girth 3-regular graphs further showed that, 
for fixed depth QAOA circuit with fixed tunable parameters,  
$O(m)$ measurements suffice to obtain a cut above the cost expectation with constant success probability \cite{farhi2026highgirth}, where $m$ is the number of edges in the graph. 
Finally, a common practice is to set the number of shots per cost evaluation  $(n)$ using a simple statistical heuristic to $n = \frac{\sigma^2}{\Delta^2}$, where $\sigma^2$ is the variance of the cost operator and $\Delta$ is the required accuracy. 
We note that this common method for setting the number of shots faces both practical and theoretical difficulties: first, the variance of the cost function is unknown and must be estimated independently; second, it nonetheless leads to a rigid statistical guarantee, in terms of the allowed confidence intervals. Yet more crucially, a key missing piece in the current analysis is the limited understanding of how the quantities of \( I \), and \(n\) scale with the size of the problems, e.g.\ with the number of edges in the graph.
Analyzing how QAOA’s total time-to-solution (TTS) scales and what fundamental factors govern this scaling, become non-trivial challenges. Addressing these questions is crucial for assessing QAOA’s practical feasibility on large problem instances.
In this paper, we investigate these questions within the context of the MaxCut problem for a graph $G=(V,E)$ with $N=\vert V \vert$ nodes and $m= \vert E \vert$ edges.
We focus on how graph size and local structural properties influence the measurement overhead and optimization complexity of QAOA, with a primary focus on the scaling of the required number of measurements. To be explicit, we ask: \textbf{given a sequence of graphs with an increasing number of edges $m$, how does the number of required shots in QAOA for MaxCut scale with $m$, and how does this scaling depend on the graph structure?}
The main contributions of this paper are as follows: \textbf{(a)} 
Using measure concentration bounds, in particular \textit{Janson’s inequality} \cite{janson_2004}, we rigorously bound the number of shots needed to estimate the QAOA cost to within a prescribed relative error $\delta>0$ and confidence level $1-\epsilon$ where $\epsilon\in(0,1)$. We show that when the cost expectation scales linearly with the number of edges $m$ (a condition we show typically holds), the shot count scales \textit{inversely} with $m$. 
Consequently, for graph sequences satisfying this extensive-cost assumption, the sufficient shot budget to achieve the same relative estimation accuracy decreases with $m$ \textbf{(Result A)}; \textbf{(b)} We derive a bound on the number of shots per cost evaluation required to ensure convergence within a stochastic gradient descent (SGD) framework. Specifically, using 
tools from optimization theory for $L$-smooth functions (i.e., functions with $L$-Lipschitz gradients), we quantify the 
sampling overhead necessary to reach a target relative stochastic error with respect to the optimal cost at a fixed circuit depth $p$. 
We then show that under reasonable assumptions on the graph structure, which are strongly linked to the behavior of the QAOA cost function landscape and its gradient (via the Polyak-Łojasiewicz (PL) constant $\mu$, described below), the number of shots scales inversely with the number of edges $m$ when using a finite-difference approximation for gradient estimation \textbf{(Result B)}; \textbf{(c)} We derive a bound on the number of SGD iterations required to achieve a target relative gap from the locally optimal cost attainable at circuit depth $p$. We show that under an additional assumption on the QAOA cost landscape smoothness (via the  Lipschitz constant $L$, described below) the number of iterations is independent of the graph size \textbf{(Result C)}; \textbf{(d)} We link our  results to existing literature analyzing the required QAOA resources (iterations and depth) for graphs that grow while preserving certain local structure, which we formalize using \textit{Benjamini-Schramm convergence}; \textbf{(e)} To assist practitioners, we translate our analytical findings into practical rules of thumb for shot-budget allocation; and (f) Finally, we validate our theoretical findings through numerical simulations.
Our observations imply that the shot count required for QAOA to maintain a constant accuracy in cost function estimation typically decreases with increasing problem size, as illustrated in Fig.~\ref{fig:schematic_illustration}. In fact, as the number of edges $m$ approaches infinity ($m\rightarrow\infty$), one can envision achieving this with as few as one circuit execution. Furthermore, for graph sequences sharing similar local structures, our results imply that the number of iterations required to achieve the same relative performance, in terms of the approximation ratio, remains constant, regardless of problem size.
Combined with prior work showing that, for certain locally tree-like graph families, fixed-depth QAOA can achieve comparable approximation ratios as graph size grows \cite{brandao2018fixed, wurtz2021fixedangle, wybo2025missing}, our findings imply a unique resource scaling for QAOA applied to MaxCut: as the system grows, the only resource requiring expansion is the gate count per layer. The number of layers and iterations remain independent of the problem size, and the shot budget per estimation actually improves. 
These results further motivate the use of QAOA for large-scale graph problems.
The rest of the paper is organized as follows: Sec.~\ref{sec:background}  provides the essential background, including the QAOA formalization for the case of the MaxCut problem and the common practice for setting the number of shots in QAOA.
Sec.~\ref{sec:analytical_results} 
then begins with key observations on the relations between graph structure, circuit depth, and approximation ratios in QAOA, and lays the specific settings and assumptions, under which we analyze the algorithm's  complexity. Then
     our central results are presented in terms of two theorems, which show that: (a) larger graphs require less measurements to achieve the same relative error of the cost estimation; (b) larger graphs require the same amount of measurements  to achieve the same relative convergence to optimal solution with the same rate. Sec.~\ref{sec:exp_results} then presents numerical results that support our analytical findings. Finally, Sec.~\ref{sec:conclusions} concludes the paper.
    \begin{figure}[H]
        \centering       
        {\includegraphics[width=0.55\linewidth]{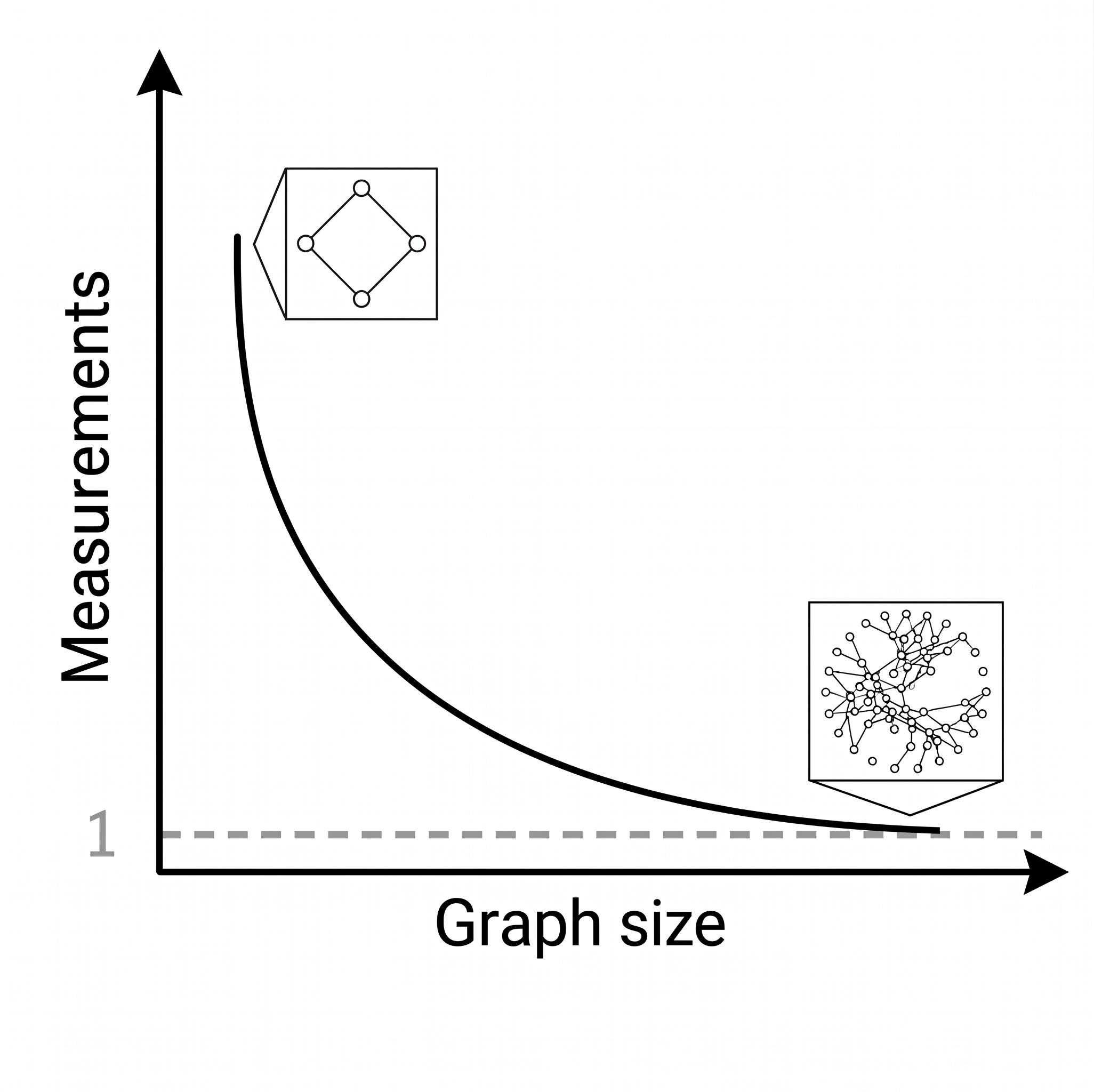}}
       \caption{
        \textbf{Schematic illustration of the main prediction developed in this paper.} Under the extensivity assumptions introduced below, larger graphs require fewer measurement shots to reach the same target relative-error within the QAOA MaxCut framework. Image generated using Gemini.}
    \label{fig:schematic_illustration}
    \end{figure}
\section{Background}\label{sec:background}
\subsection{Statistic notations and essentials}
\label{sec:statistic_background}
The following provides basic statistical notations and relations used in this work.
Given a discrete random variable $X$ with $k$ possible realizations $X_1,...X_k$, we denote its expectation value by $\mu_X \equiv \mathbb{E}[X] \equiv \langle X \rangle = \sum_{i=1}^k p_iX_i$, where $p_i$ is the probability to sample the $X_i$ realization, and denote its (population) variance by $\text{Var}(X) \equiv \sigma_X^2 \equiv  \mathbb{E}[(X-\mu_X)^2] = \sum_{i=1}^k p_i(X_i-\mu_X)^2 = \mathbb{E}[X^2]-\mathbb{E}[X]^2$, where $\sigma_X$ is denoted as the (population) standard deviation of $X$. 
Given $n$ sample realizations of $X$: $x_1,...,x_n$ we 
denote the observed sample mean by $\hat{x}_n=\frac{1}{n}\sum_{i=1}^n x_i$, and the sample variance by $s^2_{\hat{x}_n}\equiv\frac{1}{n-1}\sum_{i=1}^n \left(x_i-\hat{x}_n\right)^2$. 
For constant $a$ and $b$ it holds that $\text{Var}(aX+b) = a^2\text{Var}(X)$ and for independent random variables $X$ and $Y$ with finite variance we have $\text{Var}(X+Y) = \text{Var}(X) +\text{Var}(Y)$. Accordingly, the variance of the mean is given by $\text{Var}(\hat{x}_n)=\frac{1}{n}\text{Var}(X)$. 
The standard deviation of the mean, $\sigma_{\hat{x}_n}\equiv\sqrt{\text{Var}(\hat{x}_n)}$, often called \textit{the standard error of the mean} (SEM), estimates the deviation of the sample mean from the true expectation value. This leads to the relation 
\begin{equation}
\label{eq:standard_error}
    n = \frac{Var(X)}{\sigma_{\hat{x}_n}^2}, 
\end{equation}
which expresses the number of samples needed in order to estimate the sample mean to within at most one standard deviation of the mean  ($\sigma_{\hat{x}_n}$) away from the true expectation value, under a normal approximation this corresponds to a statistical guarantee of $\approx 68.2$ of the mass.
\paragraph{Janson's bound}
Janson's bound \cite{janson_2004} is a variant of the Hoeffding bound that enables bounds for dependent summands, which is the case in this study. Suppose that $X$ is a sum of $n$ bounded random variables $a_i \leq X_i \leq b_i$ denoted by $X=\sum\limits_{i=1}^nX_i$, the expected value of the sum is $\mathbb{E}\left[ X \right]=\mu_X$, and
the variables $X_i$ admit a dependency graph with maximum degree $\Lambda \geq 0$. That is, each $X_i$ may depend on at most $\Lambda$ other variables. Then for $t \gt 0$ it holds that  \cite{janson_2004}:
   \begin{equation}
        \label{eq:janson_bound}
       Pr\Biggl(\Bigm\lvert X-\mu_X\Bigm\lvert \ge t \Biggr) \le 2\exp \Biggl( -2\frac{t^2}{(\Lambda+1)\sum\limits_{i=1}^n(b_i-a_i)^2} \Biggr).  
   \end{equation}
Here, $\Lambda$ is not a variance; it is a combinatorial dependency parameter. The quantity $(\Lambda+1)\sum_i(b_i-a_i)^2$ plays the role of a concentration-width proxy, reducing to the usual Hoeffding denominator when the variables are independent.
In this paper, we use Janson's bound to estimate how many measurement samples are required to estimate the expectation value of the cost function to within a certain relative error $\delta$ under predefined desired statistical guarantees.
\paragraph{Hamiltonian expectation value and variance}
The expectation value of a Hamiltonian $\mathcal{H}$ depends on the system's state. 
For a parametrized circuit of $q$ qubits which prepares the pure state $\ket{\psi(\theta)}=\sum_{s=1}^{2^q}\alpha_s(\theta)\ket{s}$, where $\theta$ denotes the circuit parameters, one can define the random variable $\mathcal{H}_{\psi(\theta)}$, whose expectation value, $\mathbb{E}[\mathcal{H}_{\psi(\theta)}]$, is given by
\begin{equation*}
     \langle \mathcal{H}_{\psi(\theta)} \rangle \equiv 
    \langle \mathcal{H} \rangle_{\psi(\theta)} = \bra{\psi(\theta)}\mathcal{H}\ket{\psi(\theta)} = \sum_{s=1}^{2^q} \abs{\alpha_s(\theta)}^2\bra{s}\mathcal{H}\ket{s} = \sum_{s=1}^{2^q} p_\theta(s) \bra{s}\mathcal{H}\ket{s},
\end{equation*}
where $\ket{s}$ is the $s^{\text{th}}$ computational-basis state, and $p_\theta(s)$ is the probability of measuring the bit string $s$ in the state $\ket{\psi(\theta)}$. The variance of $\mathcal{H}_{\psi(\theta)}$ 
is defined accordingly as 
\begin{equation*}
    \text{Var}(\mathcal{H}_{\psi(\theta)}) \equiv 
    \text{Var}(\mathcal{H})_{\psi(\theta)} \equiv 
    \sigma_{H_{\psi(\theta)}}^2 \equiv \bra{\psi(\theta)}\mathcal{H}^2\ket{\psi(\theta)}- \bra{\psi(\theta)}\mathcal{H}\ket{\psi(\theta)}^2
\end{equation*}
\paragraph{Quantum statistic under a finite number of measurements}
In practice, the number of measurements $n$ is finite. In that case, we define the sample mean 
\begin{equation*}
\hat{\mathcal{H}}_{\psi(\theta)} = \frac{1}{n}\sum_{i=1}^n \bra{s_i}\mathcal{H}\ket{s_i}, 
\end{equation*}
whose variance represents the statistical noise of the sample mean, and is given by 
\begin{equation*}
Var(\hat{\mathcal{H}}_{\psi(\theta)}) = \frac{1}{n}Var(\mathcal{H}_{\psi(\theta)})
\end{equation*}
\paragraph{Partial derivative of the Hamiltonian's sample mean and its variance}
In QAOA, as in other VQAs, the tunable parameters are updated using classical optimization methods, most of which are gradient based. To that end, one seeks an estimator for the derivative of the exact expectation value with respect to a tunable parameter $\theta$, namely $\frac{\partial}{\partial \theta}\langle \mathcal{H}\rangle_{\psi(\theta)}$, constructed from finite shot sample means.
This can be approximated numerically in various ways. A simple choice is the central finite-difference approximation (second-order accurate) for the first derivative:
\begin{equation}
    \label{eq:finite_diff_hamiltonian}
    \frac{\partial \hat{\mathcal{H}}_{\psi(\theta)}}{\partial \theta}
    \approx
    \frac{\hat{\mathcal{H}}_{\psi(\theta+h)} - \hat{\mathcal{H}}_{\psi(\theta-h)}}{2h},
\end{equation}
and, more generally, a $(2k+1)$-point central finite-difference form
\begin{equation}
    \label{eq:k_point_fd}
    \frac{\partial \hat{\mathcal{H}}_{\psi(\theta)}}{\partial \theta}
    \approx
    \frac{1}{h}\sum_{j=-k}^{k} a_j \,\hat{\mathcal{H}}_{\psi(\theta + jh)},
\end{equation}
where $h$ is a finite discretization spacing and $a_j$ are the standard finite-difference coefficients for the first derivative.
A second option is the gate-wise parameter-shift rule~\cite{schuld2019evaluating} as implemented by PennyLane's \texttt{param\_shift} transform~\cite{pennylane,pennylane3}. 
PennyLane applies the shift rule at the level of trainable gate parameters: it generates shifted circuits for each such gate parameter, estimates the corresponding gate-level derivative, and then combines these derivatives. Thus, if a tunable parameter $\theta$ enters the circuit through gate parameters $\{\phi_a(\theta)\}_{a\in\mathcal A_\theta}$, then for the standard two-term shift rule of Pauli-rotation gates,
\begin{equation}
\label{eq:pennylane_gatewise_param_shift}
\frac{\partial  \hat{\mathcal H}^{\text{psr}}_{\psi(\theta)}}{\partial \theta}
=
\frac{1}{2}\sum_{a\in\mathcal A_\theta}
\frac{\partial \phi_a}{\partial \theta}
\left(\hat{\mathcal H}_{\psi_{a,+}(\theta)}
-
\hat{\mathcal H}_{\psi_{a,-}(\theta)}
\right).
\end{equation}
Here $\psi_{a,\pm}(\theta)$ denotes the circuit state obtained by shifting only the gate parameter $\phi_a$ by $\pm\pi/2$, while leaving all other gate parameters unchanged. In the QAOA MaxCut convention used here, the cost layer has $\phi_e(\gamma)=\gamma$ for each edge rotation, while the mixer layer has $\phi_i(\beta)=2\beta$ for each single-qubit $X$ rotation. Hence $\partial\phi_e/\partial\gamma=1$ and $\partial\phi_i/\partial\beta=2$.
Next, we examine the variance of the resulting derivative estimators.
For the central finite-difference estimator~\eqref{eq:finite_diff_hamiltonian} we obtain 
\begin{align}
    \label{eq:1_finite_diff_var}
    \mathrm{Var}\!\left(
        \frac{\partial \hat{\mathcal{H}}^{(1)}_{\psi(\theta)}}{\partial\theta}
    \right)
    &=
    \frac{1}{4h^{2}}
    \left(
        \mathrm{Var}\!\left(\hat{\mathcal{H}}_{\psi(\theta+h)}\right)
        +
        \mathrm{Var}\!\left(\hat{\mathcal{H}}_{\psi(\theta-h)}\right)
    \right) \\
    &=
    \frac{1}{4h^{2}n}
    \left(
        \mathrm{Var}\!\left(\mathcal{H}_{\psi(\theta+h)}\right)
        +
        \mathrm{Var}\!\left(\mathcal{H}_{\psi(\theta-h)}\right)
    \right).
    \nonumber
\end{align}
Similarly, for the general $(2k+1)$-point form~\eqref{eq:k_point_fd} we find
\begin{equation}
    \label{eq:k_finite_diff_var}
    \mathrm{Var}\!\left(
        \frac{\partial \hat{\mathcal{H}}^{(k)}_{\psi(\theta)}}{\partial\theta}
    \right)
    =
    \frac{1}{n h^{2}}
    \sum_{j=-k}^k a_j^{2}\,
        \mathrm{Var}\!\left(\mathcal{H}_{\psi(\theta+jh)}\right).
\end{equation}
For the parameter-shift estimator corresponding to \eqref{eq:pennylane_gatewise_param_shift}, we have
\begin{equation}
    \label{eq:param_shift_rule}
    \mathrm{Var}\!\left(
        \frac{\partial  \hat{\mathcal{H}}^{\mathrm{psr}}_{\psi(\theta)}}{\partial\theta}
    \right)
    =
    \frac{1}{4n}\sum_{a\in \mathcal A_\theta} \left(\frac{\partial \phi_a}{\partial \theta}\right)^2 \left(
 \mathrm{Var}\left({\mathcal{H}_{\psi_{a,+}(\theta)}}\right)
+
\mathrm{Var}\left({\mathcal{H}}_{\psi_{a,-}(\theta)}\right)\right).
\end{equation}
In all cases, the variance of the approximated partial derivative of the Hamiltonian scales linearly with the variance of the Hamiltonian up to a method-dependent prefactor, and can be bounded by
\begin{equation}
\label{eq:partial_derivative_variance}
\mathrm{Var}\!\left(\frac{\partial \hat{\mathcal{H}}_{\psi(\theta)}}{\partial\theta}\right)
\le
\frac{w_s}{n}\,
\max_{\theta \in \Theta}\mathrm{Var}\!\left(\mathcal{H}_{\psi(\theta)}\right),
\end{equation}
where $w_s$ is method dependent:
\begin{equation}  
\label{eq:ws_scaling_cases}
    w_s \;=\; 
\begin{dcases}
    \frac{1}{h^2}\sum_{j=-k}^{k}a_j^2 = O(h^{-2}),                            & \text{finite-difference}\\
        \frac{1}{2}\sum_{a \in \mathcal A_\theta} \left( \frac{\partial \phi_a}{\partial \theta}\right)^2 = \Theta(m),                            & \text{parameter-shift,}
\end{dcases}
\end{equation}
and where in the parameter-shift case, the scaling of $w_s$  with graph size corresponds to QAOA for MaxCut. We remark that using the \textit{general} parameter shift rule of Ref.~\cite{Wierichs_2022} would result in a quadratic scaling of $w_s$, i.e. $w_s^{\mathrm{gps}} =  \Omega(m^2)$, see Appendix \ref{app:parameter_shift_rule_results}.
\subsection{Optimization theory}
Next, we review basic definitions and results in optimization theory, see Ref.~\cite{wolf_mathematical} for more details. 
\paragraph{Gradient descent and stochastic gradient descent}
In gradient descent, the tunable parameters are updated according to
\begin{equation*}
\label{eq:gradient_descent}
\vec{\theta}^{\,t+1}=\vec{\theta}^{\,t}-\alpha \nabla F(\vec{\theta}^{\,t}),
\end{equation*}
where $\vec{\theta}^{\,t}$ denotes the parameters at iteration $t$, $\alpha$ is the learning rate, and $F$ is the cost function to be minimized.
In practice, for quantum algorithms $F(\vec{\theta})$ and its gradient are estimated from a finite number of measurement shots and possibly additional approximations such as finite-difference gradients. We therefore work with a random gradient estimator $\hat g^{\,t}\in\mathbb{R}^{d}$, with $d=2p$ in our parametrization, yielding the stochastic gradient descent (SGD) update rule
\begin{equation*}
\vec{\theta}^{\,t+1}=\vec{\theta}^{\,t}-\alpha \hat g^{\,t}.
\end{equation*}
Throughout, we allow $\hat g^{\,t}$ to be biased, and we define its bias vector as
\begin{equation*}
b^{t} \;\overset{\text{def}}{=}\; \mathbb{E}\!\left[\hat g^{\,t}\right]-\nabla F(\vec{\theta}^{\,t}).
\end{equation*}
When $\hat g^{\,t}$ is unbiased under an ideal parameter-shift estimation, one has $b^{\,t}=0$.
\\
\\
\paragraph{Mean-squared gradient error (MSE) and its decomposition}
A quantity that naturally appears in convergence analyses is the squared gradient error
\begin{equation}
\label{eq:mse_def}
\mathrm{MSE}^{\,t}
\;\overset{\text{def}}{=}\;
\mathbb{E}\!\left[\left\|\hat g^{\,t}-\nabla F(\vec{\theta}^{\,t})\right\|^2\right].
\end{equation}
Introduce the zero-mean fluctuation around the estimator mean,
\begin{equation*}
\varepsilon^{\,t} \;\overset{\text{def}}{=}\; \hat g^{\,t}-\mathbb{E}\!\left[\hat g^{\,t}\right],
\qquad \text{so that } \mathbb{E}[\varepsilon^{\,t}]=0.
\end{equation*}
Then
\begin{align}
\hat g^{\,t}-\nabla F(\vec{\theta}^{\,t})
&=\underbrace{\bigl(\hat g^{\,t}-\mathbb{E}[\hat g^{\,t}]\bigr)}_{\varepsilon^{\,t}}
+\underbrace{\bigl(\mathbb{E}[\hat g^{\,t}]-\nabla F(\vec{\theta}^{\,t})\bigr)}_{b^{\,t}}
=\varepsilon^{\,t}+b^{\,t}.
\end{align}
Expanding the squared norm and taking expectation gives
\begin{align}
\mathrm{MSE}^{\,t}
&=\mathbb{E}\!\left[\|\varepsilon^{\,t}+b^{\,t}\|^2\right] \nonumber\\
&=\mathbb{E}\!\left[\|\varepsilon^{\,t}\|^2\right]
+2\,\mathbb{E}\!\left[\langle \varepsilon^{\,t},b^{\,t}\rangle\right]
+\|b^{\,t}\|^2 \nonumber\\
&=\underbrace{\mathbb{E}\!\left[\|\hat g^{\,t}-\mathbb{E}[\hat g^{\,t}]\|^2\right]}_{\mathrm{Var}(\hat g^{\,t})}
+\|b^{\,t}\|^2,
\label{eq:mse_var_bias}
\end{align}
because $b^{\,t}$ is deterministic conditioned on $\vec{\theta}^{\,t}$ and
$\mathbb{E}\!\left[\langle \varepsilon^{\,t},b^{\,t}\rangle\right]
=\left\langle \mathbb{E}[\varepsilon^{\,t}],b^{\,t}\right\rangle=0$.
Equivalently, writing the estimator componentwise $\hat g^{\,t}=(\hat g^{\,t}_1,\ldots,\hat g^{\,t}_d)$, we obtain
\begin{align}
\mathrm{MSE}^{\,t}
&=\sum_{i=1}^{d}\mathbb{E}\!\left[\left(\hat g^{\,t}_i-\frac{\partial F(\vec{\theta}^{\,t})}{\partial \theta_i}\right)^2\right] \nonumber\\
&=\sum_{i=1}^{d}\Bigl(\underbrace{\mathrm{Var}(\hat g^{\,t}_i)}_{\mathbb{E}[(\hat g^{\,t}_i-\mathbb{E}[\hat g^{\,t}_i])^2]}
+\underbrace{\bigl(\mathbb{E}[\hat g^{\,t}_i]-\tfrac{\partial F(\vec{\theta}^{\,t})}{\partial \theta_i}\bigr)^2}_{{b^t_i}^2}\Bigr).
\label{eq:mse_componentwise}
\end{align}
If $\hat g^{\,t}$ is unbiased, \eqref{eq:mse_var_bias} reduces to $\mathrm{MSE}^{\,t}=\mathrm{Var}(\hat g^{\,t})$.
\\
\\
\paragraph{Connection to finite differences}
If the gradient is approximated by a $k$th-order finite-difference rule with step size $h$, then the truncation error per component scales as $O(h^{k})$, i.e.,
\begin{equation}
\label{eq:fd_bias_scaling}
\mathbb{E}[\hat g^{\,t}_i]-\frac{\partial F(\vec{\theta}^{\,t})}{\partial \theta_i} = O(h^{k}),
\end{equation}
so the squared bias term in \eqref{eq:mse_componentwise} scales as $O(h^{2k})$. Consequently,
\begin{equation}
\label{eq:mse_bound_fd}
\mathbb{E}\!\left[\left\|\hat g^{\,t}-\nabla F(\vec{\theta}^{\,t})\right\|^2\right]
=
\sum_{i=1}^{d}\mathrm{Var}(\hat g^{\,t}_i)
+\sum_{i=1}^{d}O(h^{2k}),
\end{equation}
which is the form used in our analytical bounds (with $d=2p$).
\paragraph{$L$-Lipschitz continuous gradient function}
 A function $F:\mathbb{R}^n \rightarrow \mathbb{R}$ has a $L$-Lipschitz continuous gradient on $\mathbb{R}^n$ if there exists a constant $L > 0$ such that
        \begin{equation}
        \label{eq:lipschitz_continous}
             \|\nabla F(\vec{\theta}_{1})- \nabla F(\vec{\theta}_2)\| \le L\|\vec{\theta}_1-\vec{\theta_2}\|, \, \,\forall \theta_1,\theta_2 \in \mathbb{R}^n. 
        \end{equation}
The Lipschitz constant $L$ essentially quantifies the smoothness of the function $F(\vec{\theta})$.         
Note that if a function $F(\vec{\theta})$ has an $L$-Lipschitz continuous gradient, it does not necessarily imply that it is convex. 
\paragraph{The Polyak-Łojasiewicz (PL) inequality}
A cost function $F(\vec{\theta}):\mathbb{R}^n \rightarrow \mathbb{R}$ satisfies the Polyak-Łojasiewicz (PL) inequality for some $\mu > 0$ if for all $\vec{\theta} \in \mathbb{R}^n$ it holds that
        \begin{equation}\label{eq:PL_inequality}
            \frac{1}{2} \|\nabla F(\vec{\theta})\|_2^2  = \frac{1}{2}\sum_{i=1}^{d} \left( \frac{\partial F}{\partial \theta_i} \right)^2\geq \mu \big(F(\vec{\theta}) - F(\vec{\theta}^*)\big),
        \end{equation}
where $\vec{\theta}^*$ represent the optimal parameters, for which the cost function attains its global minimum. 
Satisfying the PL inequality ensures that, whenever the gradient vanishes, the cost is at its global minimum.
Note that if a cost function $F$ with an $L\text{-Lipschitz}$ continuous gradient satisfies the PL inequality, then $F$ is invex, a weaker but more general condition than convexity, which is sufficient to guarantee global convergence \cite{wolf_mathematical}. Since our problem is neither convex nor globally invex, we instead assume local invexity, which ensures local convergence within the region of optimization.
    \paragraph{Gradient descent convergence for $L$-Lipschitz gradient cost functions that satisfy the PL-inequality}
    Let $F(\vec\theta):\mathbb{R}^d \rightarrow \mathbb{R}$ be a cost-function with an $L$-Lipschitz gradient that satisfies the PL-inequality for some $\mu \gt 0$. Then, under the gradient-descent optimization rule of Eq.~\eqref{eq:gradient_descent}
    and setting the learning rate to $\alpha=\frac{1}{L}$, it holds that
    \cite{wolf_mathematical}:
    \begin{equation}
            \forall t \in \mathbb{N} \quad F(\vec\theta^{t})-F(\vec\theta^*) \le \big( 1 - \frac{\mu}{L} \big)^{t}\Big(F(\vec\theta^0)-F(\vec\theta^*)\Big),
            \label{eq:gradient_descent_convergence}
        \end{equation}
    where $\vec\theta^{t}$ are the tunable parameters at iteration (time step) $t$ and $\vec\theta^*$ represent the optimal parameters, as before. 
        Rewriting \eqref{eq:gradient_descent_convergence} with relative form gives
        \begin{equation}
            d_t\abs{F(\vec\theta^*)} \le \big( 1 - \frac{\mu}{L} \big)^{t}d_0\abs{F(\vec\theta^*)}, 
            \label{eq:relative_gradient_descent_convergence}
        \end{equation}
        where 
        \begin{equation}
            \label{eq:d_t_relative_cost_gap}       
            d_t:=\frac{\mathbb{E}[F(\vec{\theta}^{\,t})]-F(\vec{\theta}^*)}{\abs{F(\vec{\theta}^*)}}
        \end{equation}      
        is the relative cost gap from the global optimum at iteration $t$.  
\subsection{Graph theory}
\paragraph{\BS convergence theorem}
\label{subsec:benjamini_schramm_convergence}
Let $G$ be a graph and let $u$ be a vertex in $G$. 
We denote by $B_p(u; G)$ the neighborhood of radius $p$ centered at $u$, consisting of all vertices at distance at most $p$ from $u$. We refer to $u$ as the \textit{root vertex}. 
The \BS convergence theorem describes the limiting behavior of a sequence of graphs by focusing on their local neighborhood structure at a fixed radius.
Formally, let $(G_N)_{N\ge 1}$ be a sequence of finite graphs with $\lvert V(G_N)\rvert\to\infty$. The standard vertex rooted version of \BS convergence states that for a uniformly random vertex $u_N \in G_N$ and every fixed radius $p$, the distribution of the rooted neighborhood $B_p(u_N;G_N)$ converges as $N\to\infty$. More explicitly, for every finite rooted neighborhood $h$ of radius $p$, the limit 
\[
\lim_{N\to\infty}\Pr\bigl(B_p(u_N;G_N)\cong h\bigr)
\]
exists \cite{benjamini_schramm_2001}. 
In this paper, we use an \emph{edge-rooted} variant, obtained by sampling a uniformly random edge $e_N$ and considering the radius-$p$ neighborhood around that rooted edge.
While in the most general case it is not known if a sequence of scaled graphs satisfies Benjamini–Schramm convergence, there are several known graph families for any this convergence holds. Trivial examples include cycle and grid graphs. More complex families include: (a) Random $v$-regular graphs are Benjamini-Schramm convergent, with local neighborhoods converging to those of the infinite $v$-regular tree  \cite{rahman2015localconvergence}; and 
(b) Sparse Erd\H{o}s--R\'enyi graphs in the regime $G(N,v/N)$, corresponding to $N$-nodes graphs with edge probability $v/N$, where $v$ is a constant, are Benjamini-Schramm convergent, with local limit given by a Poisson($v$) Galton--Watson tree \cite{vanderhofstadNotesRGCN}.
\subsection{QAOA for the MaxCut problem}\label{secsec:qaoa}
\paragraph{The QAOA cost function}
        The QAOA cost-function is phrased as the expectation value of a problem Hamiltonian $C$, defined per problem. 
        For a $p$-layer QAOA circuit of $q$ qubits the cost function $F_p (\boldsymbol{\gamma,\beta})$ is given by:
        \begin{equation}
        \label{eq:qaoa_expectaion} 
         F_p (\boldsymbol{\gamma,\beta})\equiv \langle  C\rangle_{\psi_p(\boldsymbol{\gamma,\beta})}  =  \bra{\psi_p(\boldsymbol{\gamma,\beta})}C\ket{\psi_p(\boldsymbol{\gamma,\beta})}, 
         \end{equation}
         with 
         \begin{equation}
         \label{eq:psi_gamma_beta}
         \ket{\psi_p(\boldsymbol{\gamma,\beta})} = U (B, \beta_p) \cdots U (C, \gamma_1) \ket{+}^{\otimes q}, 
         \end{equation}
        where $U(B, \beta_p) = e^{-i\beta_p B}$ and $U(C, \gamma_p) = e^{-i\gamma_p C}$, with $B = \sum_{i=1}^qX_i$ being the so called mixer Hamiltonian, and $\ket{+}^{\otimes q}=\frac{1}{\sqrt{2^q}}\sum_{j=1}^{2^q}\ket{j}$ is the uniform superposition state, which is an eigenstate of the mixer Hamiltonian. 
        \paragraph{QAOA formulation for the MaxCut problem}
        The MaxCut problem aims at finding an optimal partitioning of a graph $G=(V,E)$ to two complementary sets, such that the number of edges that connect nodes from the two sets, is maximized. In this case, the number of qubits is given by the number of nodes in the graph, i.e.\ $q=N$, and the problem Hamiltonian 
        
        is given by 
        \begin{equation}
        \label{eq:qaoa_problem_hamiltonian}
            C = -\!\!\!\!\sum_{\left\langle j,k \right\rangle \in E} \!\!\! C_{jk} = -\!\!\!\!\sum_{\left\langle j,k \right\rangle \in E}\frac{1}{2}\left(1-Z_jZ_k\right),
        \end{equation}
        where $C_{jk} = \frac{1}{2}\left(1-Z_jZ_k\right)$, and  $Z_k$ is the $Z$-Pauli operator applied on qubit $k$. 
        With this formulation, the groundstate of the  
        problem Hamiltonian serves as the desired solution. 
        In this case, the cost function is given by
        \begin{equation}
            F_p (\boldsymbol{\gamma,\beta})=-\!\!\!\!\sum_{\left\langle j,k \right\rangle \in E} \!\!\!\bra{\psi_p(\boldsymbol{\gamma,\beta})}C_{jk}\ket{\psi_p(\boldsymbol{\gamma,\beta})} = -\!\!\!\!\sum_{\left\langle j,k \right\rangle \in E} \!\!\!f_{p,jk}(\boldsymbol{\gamma,\beta}),
\label{eq:qaoa_expectaion_edges_sum} 
        \end{equation}  where each term 
$f_{p,jk}(\boldsymbol{\gamma,\beta})\in [0,1]$ denotes the local expectation value for edge 
$(j,k)$. The overall cost function is thus obtained by summing these local contributions across all edges in the graph.
        The approximation ratio, which expresses the relative success of the solution, is given by   
        \begin{equation}
            \label{eq:approximation_ratio}
            r_p=\frac{\abs{F_p(\boldsymbol{\gamma}, \boldsymbol{\beta})}}{C_{max}},
        \end{equation}
        where $C_{max}$ denotes the maximum cut of the graph.
        Throughout the paper, we mark the optimal parameters by $(\boldsymbol{\gamma}^*, \boldsymbol{\beta}^*)$, so that $F_p(\boldsymbol{\gamma}^*, \boldsymbol{\beta}^*)$ is the optimal solution attainable with a $p$-depth circuit, and the corresponding optimal $p$-depth approximation ratio is
        \begin{equation}
            \label{eq:opt_approximation_ratio}
            r_p^*=\frac{\abs{F_p(\boldsymbol{\gamma}^*, \boldsymbol{\beta}^*)}}{C_{max}}.
        \end{equation}
\paragraph{The QAOA Reverse Causal Cone (RCC)}
It is known that in a $p$-layer QAOA MaxCut circuit, the local expectation value of each edge $f_{p,jk}(\boldsymbol{\gamma}, \boldsymbol{\beta})$ (see Eq.~\eqref{eq:qaoa_expectaion_edges_sum}) depends only on the local topology of the graph. Specifically, it is determined by the so-called QAOA "RCC" - the subgraph consisting of all nodes and edges within a distance of 
$p$ edges from the edge $\langle j,k\rangle$  \cite{wang_quantum_2018, sureshbabu_parameter_2024, farhi2020quantumapproximateoptimizationalgorithm}. This dependence can be expressed as:
\begin{eqnarray}
        f_{p,jk}\, (\boldsymbol{\gamma}, \boldsymbol{\beta}) = && \bra{+}^{\otimes N} U^\dagger(C_{g_p(j,k)}, \gamma_1) \cdots U^\dagger (B_{g_p(j, k)}, \beta_p) \nnm
        && C_{\langle j k\rangle} U(B_{g_p(j,k)}, \beta_p) \cdots U(C_{g_p(j, k)}, \gamma_1) \ket{+}^{\otimes N},
\label{eq:qaoa_subgraph_cost}
\end{eqnarray}
where $g_p(j,k)$ denotes the $p$-radius sub-graph centered on edge $\langle j,k\rangle$. 
Hence, if two edges $\left\langle j k \right\rangle$ and $\left\langle j^\prime k^\prime \right\rangle$ give rise to isomorphic sub-graphs, then for a given choice of tunable parameters $(\boldsymbol{\gamma}, \boldsymbol{\beta})$, the corresponding edge evaluations are equal, i.e.\
        \begin{equation}
            f_{p,jk}(\boldsymbol{\gamma}, \boldsymbol{\beta})=f_{p,j^\prime k^\prime}(\boldsymbol{\gamma}, \boldsymbol{\beta}) \quad \text{if} \quad  g_p(j, k) = g_p(j^\prime, k^\prime). 
            \label{eq:qaoa_subgraph_cost2}
        \end{equation}
        We can thus view the sum in Eq.~\eqref{eq:qaoa_expectaion_edges_sum} as a sum over evaluations $f_{p,g}(\boldsymbol{\gamma}, \boldsymbol{\beta})$ of distinct sub-graph types $g$, each weighted by its corresponding multiplicity $w_g$:
        \begin{equation}
            F_p \, (\boldsymbol{\gamma}, \boldsymbol{\beta}) = -\!\sum\limits_g \, w_g\, f_{p,g} (\boldsymbol{\gamma}, \boldsymbol{\beta})
        \label{eq:qaoa_sum_over_subgraphs},
        \end{equation}
        with 
        \begin{equation}
            \label{eq:sub_over_graph_types_multiplicity}
           \sum\limits_g \, w_g = m
        \end{equation}
        and where
        \begin{equation*}
            f_{p,g}\, (\boldsymbol{\gamma}, \boldsymbol{\beta}) = f_{p,jk}\, (\boldsymbol{\gamma}, \boldsymbol{\beta}) \quad \forall \langle j,k\rangle \quad s.t. \quad g_p(j,k)=g.   
        \end{equation*}
       \paragraph{Statistical independence of edge costs from disjoint  causal cones}
       It was noted in \cite{farhi2014quantum} that for any two costs $C_{jk},C_{lm}$ corresponding to edges that belong to disjoint causal cones, i.e. their causal cones do not share any nodes, the costs are statistically uncorrelated, namely
        \begin{equation*}           \bra{\psi(\beta,\gamma)}C_{jk}C_{lm}\ket{\psi(\beta,\gamma)} = \bra{\psi(\beta,\gamma)}C_{jk}\ket{\psi(\beta,\gamma)}
            \bra{\psi(\beta,\gamma)}C_{lm}\ket{\psi(\beta,\gamma)}
        \end{equation*}
        where $\ket{\psi(\gamma,\beta)}$ is defined as in Eq.~\eqref{eq:psi_gamma_beta}. 
        This is due to the commutation of the local unitaries, namely
        \begin{eqnarray*}
             &&[e^{i\beta X_j},Z_k]=0 \quad \forall k\neq j \\  \nonumber
             &&[e^{i\gamma Z_jZ_k},Z_l]=0 \quad \forall l\neq j,k \\ \nonumber
             &&[e^{i\gamma Z_jZ_k},e^{i\gamma Z_lZ_m}]=0 \quad \forall j,k,l,m         
        \end{eqnarray*}
        The initial state is a product, $\ket+^{\otimes N}=\ket+_{L_{jk}}\otimes \ket+_{L_{lm}}\otimes \ket+_{\text{rest}}$, resulting in
        \begin{eqnarray}       
        \label{eq:uncorrelated_distinct_cost}
        && \bra{\psi(\beta,\gamma)}C_{jk}C_{lm}\ket{\psi(\beta,\gamma)} = \nonumber \\ 
        && \quad \bra{+}^{\otimes N} U_{L_{jk}} \otimes U_{L_{lm}} (Z_j \otimes Z_k \otimes  Z_l \otimes Z_m) U^{\dagger}_{L_{jk}} \otimes U^{\dagger}_{L_{lm}} \ket{+}^{\otimes N} = \\ 
        && \quad \bra{+}^{\otimes N} \left (U_{L_{jk}} (Z_j \otimes Z_k) U^{\dagger}_{L_{jk}}\right) \otimes \left(U_{L_{lm}} (Z_l \otimes Z_m) U^{\dagger}_{L_{lm}} \right)\ket{+}^{\otimes N} \nonumber
       \end{eqnarray}  
        where $U_{L_{jk}}$ involves only local operators applied on qubits that correspond to nodes in the causal cone of edge $j-k$. This implies that $U_{L_{jk}}$ and $U_{L_{lm}}$ operate on two disjoint subspaces, hence their multiplication is a tensor product. In the simple case of $p=1$ it reduces to 
        \begin{equation*}
            U_{L_{jk}}=e^{i\gamma_1(\sum_{l,m\in E_{j,k}} Z_lZ_m)} e^{i\beta_1 (X_j+X_k)}, 
        \end{equation*}
        where $\sum_{l,m \in E_{j,k}}$ represents summing over all  nodes connected either to node $j$ or $k$. For $p>1$ more qubits are involved, but still only those that belong to the j-k edge's causal cone.
        The derivation in Eq.~\eqref{eq:uncorrelated_distinct_cost} shows that not only that for edges belonging to disjoint causal cones the costs $C_{jk}$ and $C_{lm}$ are uncorrelated, but also that their measurement outcomes are statistically independent, since 
        \begin{eqnarray}
            \label{eq:qaoa_statistical_independence}
            p(a_jb_k,c_ld_m) &=& \bra{\psi(\beta,\gamma)}a_jb_k,c_ld_m\rangle \langle a_jb_k,c_ld_m\ket{\psi(\beta,\gamma)} \\ \nonumber
            &=& \bra{+}^{\otimes N} U_{L_{jk}} \otimes U_{L_{lm}} \ket{a_jb_k,c_ld_m}\bra{a_jb_k,c_ld_m} U^{\dagger}_{L_{jk}} \otimes U^{\dagger}_{L_{lm}} \ket{+}^{\otimes N} \\ \nonumber
            &=& \bra{+}_{L_{jk}} U_{L_{jk}} \ket{a_jb_k}\bra{a_jb_k}U^{\dagger}_{L_{jk}}\ket{+}_{L_{jk}}
            \bra{+}_{L_{lm}} U_{L_{lm}} \ket{c_ld_m}\bra{c_ld_m}U^{\dagger}_{L_{lm}}\ket{+}_{L_{lm}}  \\ \nonumber
            &=& p(a_jb_k)p(c_ld_m), 
        \end{eqnarray}
        where $p(a_jb_k)$ is the probability to get the outcome $a\in\{-1,1\}$ for qubit $j$ and $b\in\{-1,1\}$ for qubit $k$.
    \vspace{2mm}
       \paragraph{The variance of the QAOA cost operator of the MaxCut problem}
         For a $p$-layer QAOA circuit and for a graph with maximum degree $v$, Ref.~\cite{farhi2014quantum} derived an upper bound on the variance of the cost operator:  
        \begin{equation}
        \text{Var}(C)_{\psi_p(\boldsymbol{\gamma}, \boldsymbol{\beta})} \leqslant 2  \rvpsym \cdot m,
        \label{eq:variance_upper_bound}
        \end{equation}
        where $\rvpsym$ is the maximal local variance contribution per edge, which depends solely on $v$ and $p$ as
        \begin{equation}
        \label{eq:h_vp_definition}
             \rvpsym = 
                \begin{dcases}
                    2p+2,   & \text{if } v=2\\
                    \frac{(v-1)^{2p+2}-1}{v-2},              & \text{if } v\ge3.
                \end{dcases}
        \end{equation}
        This behavior can be understood as a consequence of locality. The total cost is a sum of $m$ local edge terms, and for fixed depth $p$ on bounded-degree graph families each term depends only on a finite causal cone. Therefore, each term can overlap with only $O(v^p)$ other terms, while the vast majority of edge pairs do not interact directly. This weak local dependence leads to extensive variance. 
\subsection{The number of shots in QAOA - the common practice}
Each measurement of the QAOA circuit results with a single computational state $\ket{s}$, i.e.\ a bit-string of length $q$, which corresponds to a cost value of $\bra{s}C\ket{s}$. 
We perform $n$ measurements, regarding each measurement as  independent sample of a random variable, calculate the sample mean $\hat{C}=\frac{1}{n}\sum_{i=1}^n \bra{s_i}C\ket{s_i}$ and ask: how many measurements  $n$ are required for $\hat{C}$ to accurately estimate the true expectation value $\langle C \rangle$?
In QAOA, as in any VQA whose cost function is phrased as the expectation value of a given Hamiltonian $C$, a common practice, see e.g. \cite{TILLY20221, zhu2023optimizing, sanders2020compilation}, is to set the number of shots $n$ to
    \begin{equation}
    \label{eq:number_of_shots_common_practice}
        n = \frac{\sigma_C^2}{\Delta^2}\equiv\frac{Var(C)}{\Delta^2},
    \end{equation}
    where $\sigma_C^2$ is the variance of the cost operator $C$, and $\Delta$ is the required accuracy of the solution, i.e.\ the required accuracy of the cost function $\langle C \rangle$ (since the variance $\sigma_C^2$ is unknown, it is typically replaced by the sample variance $\mathbf{s_C}^2$).
    This practice is based on a simple statistical heuristic, namely, to equate the \textit{standard error of the mean} (SEM), see Eq.~\ref{eq:standard_error}, with the desired error itself. The method then operates within a fixed probability guarantee (of one standard deviation). To get more flexible and tighter bounds, other statistical tools like measure concentrations, can be used, see e.g.\ Ref.~\cite{kahani2023novel} which offers a shot-allocation scheme based on the Hoeffding inequality. 
\section{Analytical results}\label{sec:analytical_results}
In this section, we present two results on the scaling of the number of shots in QAOA for the MaxCut problem. The first result (\ref{sec:analytical_result_a}) bounds the number of shots required to estimate the cost expectation. More specifically,  it answers the following question: given a graph with $m$ edges and a $p$-depth QAOA circuit, how many shots are needed to estimate the cost up to a fixed relative error?  This analysis accounts for the graph size and structure, but is agnostic to the optimization process. Intuitively, it accounts for a snapshot of the iterative procedure. 
In contrast, the second result (\ref{sec:analytical_result_b}) dives into the optimization process and provides an answer to the following question: given a graph with $m$ edges and a $p$-depth QAOA circuit, how many iterations and shots per iteration (a fixed number of shots per iteration is assumed) are needed to allow the iterative process to converge the cost function up to some relative cost gap from the optimal attainable value $F_p(\boldsymbol{\gamma}^*, \boldsymbol{\beta}^*)$ (see Eq.~\eqref{eq:opt_approximation_ratio}).
\subsection{Result A:  shot-number scaling for estimating the QAOA cost expectation} \label{sec:analytical_result_a}
In what follows, we show that for a fixed $p$-layer QAOA circuit executed on a $v$-maximal degree graph, the number of shots $n_p$ that is sufficient to estimate the cost expectation, while guaranteeing a relative error of at most $\delta > 0$ with an 
error probability $\epsilon \in (0,1)$, is lower bounded by:
\begin{equation}
    n_p \geq \frac{\eta m}{\delta^2\costsquare}, 
    \label{eq:res_a_shot_number_scale}
\end{equation}
where $\eta$ scales linearly with $\rvpsym$ (defined in Eq.~\eqref{eq:h_vp_definition}) and $\ln\left(\frac{2}{\epsilon}\right)$.  
This shows that in the most general setting, the number of shots $n_p$  that is sufficient to promise that the estimation $\hat{C_p}$ deviates by less than $\delta\langle C_p \rangle$   from its true expectation  $\langle C_p \rangle $ with probability $1-\epsilon$, increases exponentially with the depth of the circuit $p$, polynomially with the degree of the graph $v$, and linearly with the number of edges $m$, while it decays quadratically with the desired relative error $\delta$, quadratically with the cost expectation, and logarithmically with the error probability $\epsilon$.
Then, by making the reasonable assumption that at any step along the iterative process the cost expectation $\langle C_p\rangle$ is an extensive property of $m$ with some constant $\kappa$, see formal statement in assumption \ref{ass:qaoa_cost_lower_bound} below, including reasoning for its plausibility, we get that 
\begin{equation}
    n_p \ge \frac{\eta}{\delta^2\kappa^2m} = \Omega \left(\frac{1}{m}\right),  
    \label{eq:inverse_relation_shots_m}
\end{equation}
which implies that as long as the graph scaling is not pathological and satisfies our assumption~\ref{ass:qaoa_cost_lower_bound}, the lower bound on the number of required shots $n_p$ scales \textit{inversely} with the number of edges $m$.
Notably, for $v$-regular graphs, where the number of edges scales linearly with the number of nodes ($m = \frac{Nv}{2}$), the required shots also scale inversely with the node count, yielding $n_p = \Omega(1/N)$.
\begin{assumption}[Extensive lower bound on the QAOA cost]
    \label{ass:qaoa_cost_lower_bound}
    Let $G$ be a graph with $m$ edges, and let $\vec{\theta}\in\mathbb{R}^{2p}$ denote the QAOA tunable parameters. Then there exists a constant $\kappa>0$, independent of $m$, such that
\begin{equation}
    \label{eq:C_lower_bound}
\costexpect \equiv \big\lvert \FpThetaG \big\rvert \;\ge\; \kappa\, m .
\end{equation}
    
    \noindent
    Reasoning for Assumption \ref{ass:qaoa_cost_lower_bound}: in Appendix~\ref{app:cost_function_scaling}.
\end{assumption}
Next, we provide two derivations for result A, as stated in Eq.~\eqref{eq:res_a_shot_number_scale}: first, in Observation~\ref{thm:shots_based_on_common_practice} using the common practice of Eq.~\eqref{eq:number_of_shots_common_practice}; and second, in Thm.~\ref{thm:shots_per_iteration} based on concentration measures in the form of  Janson's bound from Eq.~\eqref{eq:janson_bound}.
The latter yields the same scaling with a tighter bound, and allows the error probability to be set more flexibly.
\subsubsection{Using the common practice heuristic}
\begin{observation}
\label{thm:shots_based_on_common_practice}
Let $\langle C_p \rangle=F_p(\vec\theta^t)$ be the expectation value of the QAOA cost operator with fixed $p$-depth circuit parameters $\vec\theta \in \mathbb{R}^{2p}$, and let $\delta \gt 0$ be a relative error. Then, within the common practice heuristic, a sufficient number of shots $n_p$ to estimate $\langle C_p \rangle$ to within error $\Delta = \delta \costexpect$ is given by
\begin{equation}
    n_p \geq \frac{\eta_{vp} m}{\delta^2 \costsquare}.
    \label{eq:n_common_practice_bound_non_tight}
\end{equation}
where $\eta_{vp} = 2\rvpsym$ with $\rvpsym$ defined in Eq.~\eqref{eq:h_vp_definition}.
\end{observation}
\begin{proof}[Derivation]
Taking the common practice of Eq.~\eqref{eq:number_of_shots_common_practice}, plugging in the variance upper bound from Eq.~\eqref{eq:variance_upper_bound} and the relative error $\Delta = \delta \costexpect$ we get a lower bound on the number of shots $n_p$ 
    \begin{equation}
        \frac{\text{Var} \left( C_p \right)}{\Delta^2}\le \frac{2\rvpsym m}{\delta^2 \costsquare} \leq n_p.
        \label{eq:n_common_practice_bound} 
    \end{equation}
\end{proof}
    \noindent
    Further applying our assumption~\ref{ass:qaoa_cost_lower_bound} that the cost scales linearly with the number of edges (Eq.~\eqref{eq:C_lower_bound}), leads to
    \begin{equation}
        n_p \geq \frac{2\rvpsym}{\delta^2 \kappa^2 m}.    
    \end{equation}
\subsubsection{Using measure concentration tools}
Next, we employ measure concentration tools, in the form of Janson's bound from Eq.~\eqref{eq:janson_bound}, to achieve tighter and more general bounds on the number of required shots. 
\begin{theorem}
        \label{thm:shots_per_iteration}
        Let $\langle C_p \rangle=F_p(\vec\theta^t)$ be the expectation value of the QAOA cost operator with fixed $p$-depth circuit parameters $\vec\theta \in \mathbb{R}^{2p}$, and $\delta \gt 0$ is a desired relative error.
        In addition, let $\epsilon\in (0,1)$ be an error probability. 
        Then, the number of shots $n_p$ required to guarantee
        \begin{equation*}
          \Pr \left(\Bigm \lvert \hat{C}_p  - \langle C_p \rangle\Bigm\lvert \gt \delta \costexpect\right) \le \epsilon,  
        \end{equation*}
        is
        \begin{equation}
             n_p \ge \frac{\eta_{vp\epsilon}m}{\delta^2\costsquare},
             \label{eq:n_variance_bound_non_tight}
        \end{equation}  
        where 
        $\eta_{vp\epsilon} = (\rvpsym + \frac{1}{2})
        \ln(\frac{2}{\epsilon})$
        with $v$ denoting the maximal degree, and $\rvpsym$ is defined in Eq.~\eqref{eq:h_vp_definition}. 
\end{theorem}
\begin{proof}
    We determine the number of shots $n_p$ that allows estimating the expectation value $\langle C_p \rangle$ to within an error of $\delta \langle C_p \rangle$, with an error probability of at most $\epsilon$, using Janson's bound from Eq.~\eqref{eq:janson_bound}. To that end, we first prove the following lemma:
    \begin{lemma}
        \label{lem:qaoa_concentration_inequality}
        Let $\langle C_p \rangle$ be the expectation of the cost for a $p$-depth QAOA circuit, and let $\langle \hat{C}_p \rangle$ denote its estimator. Let $m$ be the number of edges in the graph, $n_p$ the number of circuit measurements, and $\rvpsym$ the bound from Eq.~\eqref{eq:h_vp_definition}. Then, for any $\delta > 0$, we have:
        \begin{equation}
        \Pr \bigl(\Bigm \lvert \hat{C}_p - \langle C_p \rangle\Bigm\lvert \gt \delta \costexpect\bigr) \le 2\exp\left(-\frac{2n_p\delta^2 \langle C_p \rangle^2}{(2\rvpsym + 1)m} \right).\label{eq:variance_bound_qaua}
        \end{equation}
    \end{lemma}
    \begin{proof}
    of lemma \ref{lem:qaoa_concentration_inequality}: in Appendix~\ref{appendix:qaoa_concentration_bound}.
    \end{proof}
    To ensure that the probability of a large estimation error on the left-hand side is sufficiently small (say, smaller than $\epsilon$), we require the right-hand side to be less than $\epsilon$. This yields the required lower bound as stated in Eq.~\eqref{eq:n_variance_bound_non_tight}:
    \begin{equation}
    \nonumber
        n_p \ge \frac{\eta_{vp\epsilon}m}{\delta^2 \langle C_p \rangle^2}, 
        \label{eq:n_variance_bound}
    \end{equation}
    with $\eta_{vp\epsilon} = \left(\rvpsym+\frac{1}{2}\right)\ln\left(\frac{2}{\epsilon}\right)$ and $\rvpsym$ defined in Eq.~\eqref{eq:h_vp_definition}, which completes the proof of Thm.~\ref{thm:shots_per_iteration}. 
    \end{proof}
As a remark, note that the lower bound of Eq.~\eqref{eq:n_variance_bound_non_tight}, obtained with Janson's bound, is tighter than the one obtained using the common practice in   Eq.~\eqref{eq:n_common_practice_bound_non_tight}. To that end, we  fix the same success probability of normal distribution $1-\epsilon \approx 0.68$, and get that
\begin{equation*}
\eta_{vp\epsilon} = 
    \left(2 \rvpsym + 1\right) \frac{\ln(\frac{2}{0.32})}{2} 
    \approx 
 (2 \rvpsym + 1)\cdot0.916 <  2 \rvpsym = \eta_{vp}, 
\end{equation*}
under 
$\rvpsym \ge 6$, which holds for all $v\geq3$ or $p \geq 2$. 
    Further plugging in our assumption~\ref{ass:qaoa_cost_lower_bound} that the cost scales linearly with the number of edges (Eq.~\eqref{eq:C_lower_bound}) gives 
    \begin{equation}
    \label{eq:inverse_scaling_measure_conecentration}
         n_p \ge 
        \frac{\eta_{vp\epsilon}}{\delta^2\kappa^2 m},
    \end{equation}
    in accordance with Eq.~\eqref{eq:inverse_relation_shots_m}. 
\subsubsection{Origin of the $1/m$ shot-count scaling}
The inverse $1/m$ scaling emerges because: (a) we employ a  target \textit{relative} rather than an absolute performance metric, and (b) both the cost expectation $\langle C_p \rangle$ and cost variance $\text{Var}(C_p)$ scale linearly with the number of edges $m$. Together, these factors imply a favorable \textit{relative concentration} of the cost function: 
since the standard deviation $\sigma_{C_p}$ grows only as $\sqrt{m}$, it becomes a vanishing fraction of the $O(m)$ expectation. Consequently,
the standard deviation of the relative cost $\sigma(C_p/\langle C_p \rangle)=\sigma(C_p)/\langle C_p \rangle$ scales as $1/\sqrt{m}$,
allowing the required shot budget to decrease while maintaining a fixed relative precision. Specifically, employing the common practice logic of Eq.~\ref{eq:number_of_shots_common_practice}, we get $\shots = Var(C_p/\langle C_p \rangle)\cdot \frac{1}{\delta^2} \propto 1/m$ for a fixed relative error $\delta$.
\begin{figure}[H]
    \centering
    \includegraphics[width=1\linewidth]{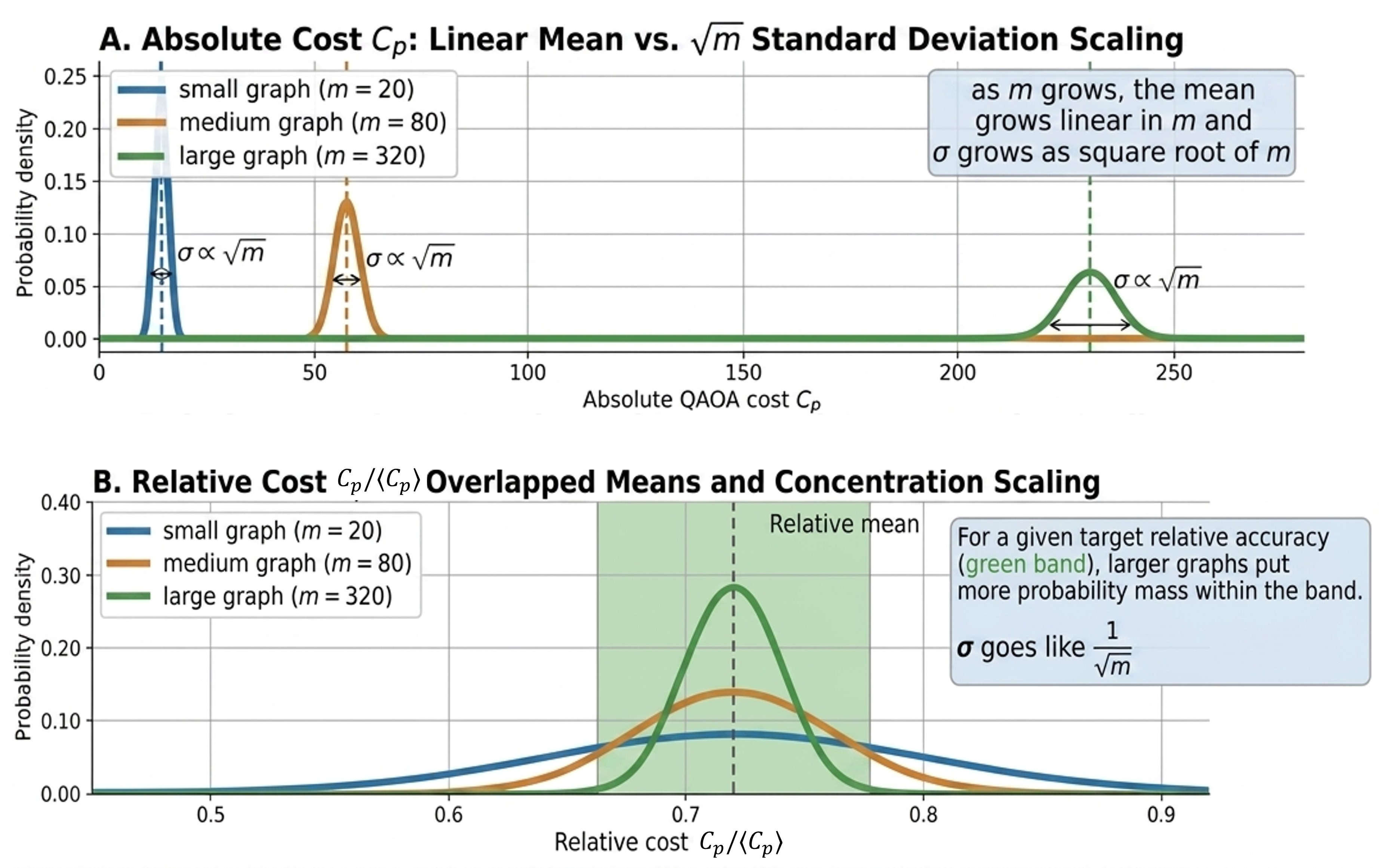}
    \caption{Schematic illustration of the concentration mechanism underlying the shot scaling results. At fixed depth $p$ and under our extensivity assumptions, the expectation $\langle C_p \rangle$ scales linearly with the number of edges $m$, while the standard deviation $\sigma(C_p)$ scales only as $\sqrt{m}$. Consequently, the relative cost $C_p/\langle C_p \rangle$ becomes more concentrated as $m$ increases, with $\sigma(C_p/\langle C_p \rangle)\propto 1/\sqrt{m}$, so fewer shots suffice to achieve the same relative estimation accuracy on larger graphs. Image generated using Gemini.}  \label{fig:concentration_illustration}
\end{figure}
\subsection{Result B: shot-number scaling along a stochastic gradient descent QAOA optimization}
\label{sec:analytical_result_b}
The result of Sec.~\ref{sec:analytical_result_a} provides an effective way of determining the number of shots needed for estimating the QAOA cost. In practice, however, QAOA is executed inside a classical optimization loop, where the cost gradient is evaluated repeatedly.
To determine the overall resource requirements, we perform a deeper analysis on the convergence of the iterative procedure, while accounting for different gradient-estimation (differentiation) methods of the cost function.
At each iteration, calculating the derivative per  tunable parameter ($2p$ parameters in total) requires the evaluation of the cost expectation at several points in the parameter space (depending on the differentiation procedure). In what follows, we analyze the number of shots required per such cost evaluation, when the parameters are updated using \textit{stochastic gradient descent} (SGD). For simplicity, we assume a fixed number of shots per cost evaluation. 
We start with Thm.~\ref{thm:shots_per_iteration_sgd} which gives a sufficient condition on the number of shots required to achieve a desired convergence accuracy, after a sufficiently large number of iterations.   Specifically, we assume that the iterative SGD procedure has reached the so-called \textit{\textbf{steady-state regime}}, where the parameters no longer converge but rather fluctuate around a local minimizer $\vec{\theta}^*$, and derive a bound on the number of shots per cost evaluation $\shots$ 

that is sufficient to limit the magnitude of these fluctuations to a desired level $\xi$.

\begin{theorem}
\label{thm:shots_per_iteration_sgd}
    Let $\costptheta$ denote the expected cost of a depth-$p$ QAOA circuit. Assume that $\costptheta$ satisfies the PL inequality (Eq.~\eqref{eq:PL_inequality})  for some $\mu>0$ on the local optimization region reached by the SGD trajectory, has an $L$-Lipschitz continuous gradient where $\mu \le L$, and satisfies Assumption \ref{ass:qaoa_cost_lower_bound}. 
    Then, for an unbiased estimator in the steady state regime (iteration $t \rightarrow \infty$),  a sufficient number of shots per cost evaluation which ensures that the relative cost gap, $d_t$ (defined in Eq.~\eqref{eq:d_t_relative_cost_gap}), is bounded by a target relative error $\xi^*$,
    i.e.\ $d_t \leq \xi^*$, is given by 
    \begin{equation}
    \label{eq:sufficient_shots_per_cost_evaluation}
        \shots \;\ge\; \frac{\varphi}{\xi^* \, \kappa \, \mu},
    \end{equation}
    where $\varphi=\varphiequal$ depends on the choice of the partial derivative approximation through $w_s$ (defined in Eq.~\eqref{eq:ws_scaling_cases}), and where $\kappa$ defined in Appendix~\ref{app:cost_function_scaling}.
\end{theorem}
\begin{proof}
Our analysis starts from the following Lemma~\ref{lem:relative_sgd_convergence}, which can be viewed as a stochastic counterpart of the relative gradient descent convergence theorem of  Eq.~\eqref{eq:relative_gradient_descent_convergence}, and which makes the dependence on $\mathrm{Var}\bigl(\gttheta\bigr)$ explicit.
For readability, we write $F(\vec{\theta}) \equiv \costptheta$ and keep the depth-$p$ dependence implicit.
\begin{lemma}
\label{lem:relative_sgd_convergence}
    Let $F(\vec{\theta})$ be a cost function with an $L$-Lipschitz gradient that satisfies the PL-inequality for some $\mu > 0$ (see Eq.~\eqref{eq:PL_inequality}). Consider SGD with learning rate $\alpha = 1/L$. Then, any iteration $t$ satisfies
    \begin{equation}
        d_t \,\bigl\vert F(\vec{\theta}^*)\bigr\vert
        \;\leq\;
        d_0\left(1-\frac\mu L\right)^t
        \bigl\vert F(\vec{\theta}^*)\bigr\vert
        \;+\;
        \xi \,\bigl\vert F(\vec{\theta}^*)\bigr\vert,
        \label{eq:relative_sgd_converge_theorem}
    \end{equation}
    where $d_t$ is the relative cost gap at iteration $t$, as defined in Eq.~\eqref{eq:d_t_relative_cost_gap}, and $\xi$ is the \emph{relative stochastic error}, defined as 
    \begin{equation}
       \xi \;\overset{\text{def}}{=}\; \frac{\sigma^2} 
       {2\mu \,\bigl\vert F(\vec{\theta}^*)\bigr\vert},
       \label{eq:xi_defnition}
    \end{equation}
where $\sigma^2$ is an upper bound on the mean-squared error (MSE) of the gradient estimate
\begin{equation}
\label{eq:maximal_MSE_of_grad_estimate}
    \mathbb{E}\!\left[\bigl\|\hgt-\nabla F(\vec{\theta^{\,t}})\bigr\|^2\right] \leq \sigma^2 \quad \forall t.
\end{equation}
\end{lemma}
\begin{proof}[Proof of Lemma~\ref{lem:relative_sgd_convergence}]
    See Appendix~\ref{app:lem_relative_sgd_convergence_proof}.
\end{proof}
\textcolor{black}{
\begin{corollary}
\label{cor:iterations_for_steady_state}
Under the assumptions of Lemma~\ref{lem:relative_sgd_convergence}, for any $\varepsilon>0$
there exists $T_\varepsilon$ such that for all $T\ge T_\varepsilon$,
\begin{equation}
d_T \le \varepsilon + \xi.
\end{equation}
\end{corollary}
}
Next, to prove 
Thm.~\ref{thm:shots_per_iteration_sgd} we express $\sigma^2$, as defined in Eq.~\eqref{eq:maximal_MSE_of_grad_estimate}, in terms of the number of shots per cost evaluation $\shots$. To that end, note that the squared gradient error $\mathbb{E}\!\left[\bigl\|\hgt-\nabla F(\vec{\theta^{\,t}})\bigr\|^2\right]$ combines both the variance of the estimator and (possibly) a finite-difference truncation error (see Eq.~\eqref{eq:mse_componentwise}):
\begin{equation}
    \mathbb{E}\!\left[\|\hgt-\nabla F(\vec{\theta}^{\,t})\|^2\right] = \sum_{i = 1}^{2p}\left(\mathrm{Var}\!\left(\frac{\partial \hat{F}^{\,t}}{\partial \theta_i}\right)+{b^t_i}^2\right),
\end{equation}
In what follows, we account only for the error originating from the variance of the estimator, assuming an unbiased calculation, namely $b^t_i=0$, and analyze the truncation error term immediately after. 
We now derive an upper bound $\sigma^2$ on the (unbiased) gradient MSE using the upper bound on the variance of the cost derivative, stated in Eq.~\eqref{eq:partial_derivative_variance}, and the upper bound on the cost variance, stated in Eq.~\eqref{eq:variance_upper_bound}:
\begin{equation}
\label{eq:gradient_variance_bound}
\begin{aligned}
    \sum_{i = 1}^{2p}\mathrm{Var}\!\left(\frac{\partial \hat{F}^{\,t}}{\partial \theta_i}\right) &\leq 2p \max_{t,i} \mathrm{Var}\!\left(\frac{\partial \hat{F}^{\,t}}{\partial \theta_i}\right)
    \leq 2p \frac{w_s}{\shots} \max_{\vec{\theta}} \mathrm{Var}\bigl(F(\vec{\theta})\bigr) \\
    &\leq\frac{4p\,w_s\,c_{vp}\,m}{\shots} 
    =\frac{2\,\varphi\,m}{\shots}, 
\end{aligned}
\end{equation}
which allows us to replace $\sigma^2$ with the bound $\frac{2\varphi m}{\shots}$ in Eq.~\eqref{eq:xi_defnition}, leading to
 \begin{equation}
 \label{eq:relative_stochastic_error_final}
     \xi = \frac{\sigma^2} 
       {2\mu \,\bigl\vert F(\vec{\theta}^*)\bigr\vert}
       =
       \frac{\varphi m}{\shots \mu \,\bigl\vert F(\vec{\theta}^*)\bigr\vert}. 
 \end{equation}
From Corollary~\ref{cor:iterations_for_steady_state}, in the steady-state limit $t\to\infty$, a sufficient number of shots per cost evaluation which ensures that the noise floor satisfies $\xi\le \xi^*$ is
\begin{equation}
\shots \geq \frac{\varphi m}{\xi^*\mu \bigl\vert F(\vec{\theta}^*)\bigr\vert}.
\end{equation}
Finally, we employ our Assumption \ref{ass:qaoa_cost_lower_bound} that the cost function scales linearly with $m$ and obtain a sufficient condition on the number of shots per cost evaluation 
    \begin{equation*}
         \shots \ge  \frac{\varphi}{\xi^*
        \,\kappa\,\mu}.
    \end{equation*}
     which completes the proof of Thm.~\ref{thm:shots_per_iteration_sgd}.
    \end{proof}
\begin{remark}
    The bound above uses the same statistical ingredient that underlies the common-practice heuristic of Eq.~\eqref{eq:number_of_shots_common_practice}: the variance of a sample mean scales inversely with the number of shots. Here, this principle is applied to each cost evaluation entering the derivative estimator, which leads to the factor $\frac{w_s}{\shots}\max_{\vec{\theta}}\mathrm{Var}\bigl(F_p(\vec{\theta})\bigr)$ in the gradient variance bound of Eq.~\eqref{eq:gradient_variance_bound}.
    An alternative route would be to use measure concentration bounds to obtain probabilistic convergence guarantees.
\end{remark}
\begin{remark}
Note that combining Eq.~\eqref{eq:gradient_variance_bound} and Eq.~\eqref{eq:ws_scaling_cases} implies that the variance of the per-parameter derivative estimator scales as
\begin{equation}  
\label{eq:var_grad_bound_scaling_cases}
    \mathrm{Var}\!\left(\frac{\partial \hat{F}^{\,t}}{\partial \theta_i}\right) \;=\; 
\begin{dcases}
     \mathcal{O}(m),                            & \text{finite-difference}\\
        \mathcal{O}(m^2),       & \text{parameter-shift,}
\end{dcases}
\end{equation}
for a fixed shot count $\shots$ per (shifted) cost evaluation.
\end{remark}
To turn Eq.~\eqref{eq:sufficient_shots_per_cost_evaluation} into a graph-size scaling statement we note that the shot condition in Eq.~\eqref{eq:sufficient_shots_per_cost_evaluation} depends on the graph size through two quantities: the PL constant $\mu$ and the estimator-dependent prefactor $\varphi$. We first state a scaling assumption on $\mu$ and then distinguish the  parameter-shift and finite-difference estimators.

\vspace{-5mm}
\begin{assumption}[Linear scaling of the QAOA PL constant $\mu$]
    \label{ass:linear_scaling_mu}
    The expected QAOA cost at depth $p$, $F_p(\vec{\theta})$,  satisfies the PL inequality (Eq.~\eqref{eq:PL_inequality}) with some $\mu > 0$  which scales linearly with the size of the graph:
    \begin{equation}
        \mu = \Theta(m).
    \end{equation}
Reasoning for  Assumption~\ref{ass:linear_scaling_mu}: provided in Appendix~\ref{app:linear_scaling_mu_L}.
\end{assumption}
For the parameter-shift rule, the estimator is unbiased, so the bias term $b^t_i$ is absent and Eq.~\eqref{eq:sufficient_shots_per_cost_evaluation} holds as is. However as discussed in Sec.~\ref{sec:statistic_background}, the number of shifted cost evaluations scales linearly with the graph size in QAOA MaxCut, and hence $\varphi=\Theta(m)$. Together with $\mu = \Theta(m) $, Eq.~\eqref{eq:sufficient_shots_per_cost_evaluation} gives $\shots = \Omega(1)$, meaning that the shot count per cost evaluation is independent of graph size.

For finite difference, once the step size $h$ and approximation order $k$ are fixed, the number of shifted cost evaluations is independent of $m$ (see Eq.~\eqref{eq:ws_scaling_cases}). Hence $\varphi=O(1)$. Applying Assumption~\ref{ass:linear_scaling_mu} to the shot-noise condition in Eq.~\eqref{eq:sufficient_shots_per_cost_evaluation} therefore gives $\shots=\Omega(1/m)$ for the shot-controlled part.
The finite-difference estimator also has a deterministic truncation bias. Since, at fixed depth $p$ and bounded degree, the QAOA cost is a sum of $m$ local edge contributions, this bias can scale as $O(mh^k)$ per gradient component. Its squared contribution to the gradient MSE is therefore $O(m^2h^{2k})$ and the corresponding contribution to the relative error floor in Eq.~\eqref{eq:xi_defnition} with Assumptions~\ref{ass:qaoa_cost_lower_bound} and~\ref{ass:linear_scaling_mu} scales as
\begin{equation}
O\!\left(\frac{m^2h^{2k}}{\mu \lvert F(\vec{\theta}^*)\rvert}\right)
=
O\!\left(\frac{m h^{2k}}{\mu}\right)=O(h^{2k}).
\end{equation}
This term is controlled by the finite-difference step size, not by the number of shots. The estimate above shows that, after normalizing by $\mu \lvert F(\vec{\theta}^*)\rvert$, the possible $m^2$ growth of the squared absolute truncation error gets canceled by the extensive cost scale and by $\mu=\Theta(m)$, leaving an $O(h^{2k})$ contribution that does not grow with the graph size. Combining this term with the shot-noise contribution gives, for finite difference approximation,

\begin{equation}
\label{eq:finite_difference_xi_with_bias}
    \xi_{\mathrm{fd}}
    \;\le\;
    \frac{\varphi}{\shots \mu \,\kappa}
    + \xi_h,
    \qquad
    \xi_h=O(h^{2k}).
\end{equation}
Here $\xi_h$ denotes a nonnegative upper bound on the normalized truncation-bias contribution. Thus, for fixed depth $p$, fixed finite-difference order $k$, bounded degree, and a fixed graph-family scaling regime, $h$ only has to be chosen small enough that this deterministic error floor is below the target tolerance. If $\xi_h<\xi^*$, then the shot-noise term only has to satisfy $\frac{\varphi}{\shots\mu\kappa}\le \xi^*-\xi_h$. Once such a step size has been calibrated on representative smaller graphs, the same $h$ can be kept for larger graphs in the same regime. Increasing $m$ changes the shot-noise requirement but does not amplify this normalized truncation-bias contribution. 
Therefore, for a fixed finite-step $h$ the number of shots required to ensure a desired target relative error $\xi^*$ scales inversely with the number of edges $m$, i.e.\ $\shots=\Omega(1/m)$.
Overall, we get
\begin{equation}  
\label{eq:shots_per_case}
    \shots \;=\; 
\begin{dcases}
    \Omega\!\left(1/m\right),                            & \text{finite-difference}\\
    \Omega\!\left(1\right),                            & \text{parameter-shift}.
\end{dcases}
\end{equation}
\subsection{Result C: Iteration scaling in QAOA with stochastic gradient descent}
\label{sec:analytical_result_c}
So far, we have examined how the number of shots per cost evaluation depends on the graph size, but have not yet analyzed how the total number of iterations $T$ scales with it. 
The following lemma provides a sufficient condition for $T$ to achieve a target relative cost gap $d^*$. Taken under Assumptions~\ref{ass:linear_scaling_mu}-\ref{ass:linear_scaling_L}, this lemma shows that the iteration count $T$ is, in fact, independent of the graph size.
\begin{theorem}
    \label{thm:req_iteration_number}
    Let $d^* \in (\xi,1)$ be the desired relative cost gap from the local optimum $F(\vec{\theta}^*)$, let $\xi$ denote the relative stochastic error term as in Eq.~\eqref{eq:xi_defnition}, and let $r_p^*$ be the best possible approximation ratio at circuit depth $p$ (see Eq.~\eqref{eq:opt_approximation_ratio}). Then a sufficient condition on the total number of iterations $T$ to ensure $d_T \le d^*$ is
    \begin{equation}
        T \;\ge\;
        \log_{1-\frac{\mu}{L}}\!\left(
            \frac{(d^* - \xi)\, 2 r_p^*}{2 r_p^* - 1}
        \right).
        \label{eq:total_iterations_lower_bound}
    \end{equation}
\end{theorem}
\begin{proof}
    Let $q:=1-\frac{\mu}{L}\in(0,1)$. Applying Eq.~\eqref{eq:relative_sgd_converge_theorem} at iteration $T$ gives
            \begin{equation}
                d_T \le q^T d_0 + \xi,
            \end{equation}
    Therefore, a sufficient condition for $d_T\le d^*$, with $d^*>\xi$, is
            \begin{equation}
                q^T d_0 + \xi \le d^*.
            \end{equation}
    which can be rearranged into
            \begin{equation}
                T \ge \log_{q}\!\left(\frac{d^*-\xi}{d_0}\right).
                \label{eq:qaoa_iteration_count}
            \end{equation}
    Next, using $F(\vec{\theta}^*)=-r_p^* C_{\max}$, we write
    \begin{equation}
        d_0=\frac{F(\vec{\theta}^{(0)})-F(\vec{\theta}^*)}{\abs{F(\vec{\theta}^*)}} = \frac{F(\vec{\theta}^{(0)})/C_{max}-F(\vec{\theta}^*)/C_{max}}{\abs{F(\vec{\theta}^*)}/C_{max}} = 
        \frac{F(\vec{\theta}^{(0)})/C_{max}+r_p^*}{r_p^*}
    \end{equation}
    and using the result that $F_p(\vec\theta^0) \leq -\frac{m}{2}\leq -\frac{C_{max}}{2}$ from \eqref{eq:qaoa_cost_m_div_2}, we arrive at
    \begin{equation}
        d_0 \leq \frac{r_p^*-\frac{1}{2}}{r_p^*}, 
    \end{equation}
    which finally gives
    \begin{equation}
            T \ge log_{1-\frac{\mu}{L}}\left(\frac{(d^*-\xi) 2r_p^*}{2r_p^*-1}\right), 
    \end{equation}
    as required.           
    \end{proof}
\begin{assumption}[Linear scaling of the Lipschitz constants $L$]
    \label{ass:linear_scaling_L}
    The expected QAOA cost at depth $p$, $F_p(\vec{\theta})$,  satisfies the PL inequality (Eq.~\eqref{eq:PL_inequality}) and has an $L$-Lipschitz continuous gradient which scales linearly with the size of the graph:
    \begin{equation}
        L = \Theta(m).
    \end{equation}
Reasoning for Assumption~\ref{ass:linear_scaling_L}: in Appendix~\ref{app:linear_scaling_mu_L}.
\end{assumption}
\begin{corollary}
    Under Assumptions~\ref{ass:linear_scaling_mu}-\ref{ass:linear_scaling_L}, we have $\mu/L = \Theta(1)$ as $m$ grows. Consequently, the lower bound on $T$ in Eq.~\eqref{eq:total_iterations_lower_bound} depends solely on parameters that remain constant relative to the graph size. This yields 
    \begin{equation}
        T = \Theta(1), 
    \end{equation}
    meaning the number of iterations required for convergence is asymptotically independent of the graph size.
\end{corollary}
    Together with the shot scaling $\shots = \Omega(1/m)$ derived for stochastic gradient descent in Eq.~\eqref{eq:shots_per_case} using the finite-difference approximation, the result $T = \Theta(1)$ achieved under Assumptions~\ref{ass:linear_scaling_mu}-\ref{ass:linear_scaling_L}, shows that larger graphs require the same number of iterations but \textit{fewer} shots per cost evaluation. Consequently, for a fixed $p$-depth QAOA circuit, achieving a target relative performance for the MaxCut problem actually becomes more shot-efficient as the graph size increases, reducing the total number of shots required across the entire optimization process.
To conclude, while Sec.~\ref{sec:analytical_result_a} establishes the scaling of the number of shots required to estimate the cost itself via Eq.~\eqref{eq:n_variance_bound_non_tight}, Sec.~\ref{sec:analytical_result_b} determines the scaling of the number of shots per cost evaluation needed for accurate gradient estimation via Eq.~\eqref{eq:sufficient_shots_per_cost_evaluation}, and Sec.~\ref{sec:analytical_result_c} derives the scaling of the total iterations required for convergence via Eq.~\eqref{eq:total_iterations_lower_bound}. Together, 
these results clarify the scaling of total shot complexity with graph size for solving MaxCut via QAOA.
\subsection{The local structure perspective} 
We briefly discuss structural settings under which Assumptions~\ref{ass:qaoa_cost_lower_bound},\ref{ass:linear_scaling_mu},\ref{ass:linear_scaling_L} are plausible, and when this leads to an inverse scaling between the number of edges and the number of shots required for cost estimation.
From a structural standpoint, a convenient sufficient condition supporting all three assumptions is \BS convergence of the graph sequence, which guarantees that the empirical frequencies of rooted radius-$p$ neighborhoods converge to a limiting distribution. 
This local-limit viewpoint is well established in classical combinatorial optimization, where local algorithms can be analyzed through the limiting local structure and yield provable cut densities in locally tree-like regimes \cite{lyons2016factorsiidtrees}. Closely related ideas appear in the QAOA literature for random regular and sparse \ER graphs: for fixed depth $p$, the per-edge objective contribution (see Eqs.~\eqref{eq:qaoa_subgraph_cost}-\eqref{eq:qaoa_subgraph_cost2}) concentrates, such that a fixed depth $p$ suffices with fixed optimal parameters, which are independent of graph size \cite{brandao2018fixed,streif2019treeqaoa,wurtz2021fixedangle,wybo2025missing}.
In such settings it is sufficient to optimize the QAOA parameters on small representative instances and transfer these angles to much larger instances without repeated outer-loop optimization, as demonstrated in the above references. 
Combined with our shot-scaling result under Assumption~\ref{ass:qaoa_cost_lower_bound}, we find that larger instances not only circumvent the need for an outer optimization loop but also require fewer shots to estimate the cost to a fixed relative accuracy and confidence level, thus offering a clear scaling benefit.
Finally, \BS convergence is sufficient but not necessary for our conclusions. Even when a sequence does not strictly admit a local weak limit, the inverse scaling between the number of edges and the required number of shots persists whenever Assumptions~\ref{ass:qaoa_cost_lower_bound}, \ref{ass:linear_scaling_mu}, and \ref{ass:linear_scaling_L} are satisfied.
 \subsection{Setting the number of shots per cost evaluation: practitioner heuristics}
Our analytical bounds can be translated into a simple, conservative recipe for determining the number of measurement shots required for each evaluation of the QAOA cost. The key inputs are the circuit depth $p$, the chosen gradient-estimation method, and a target optimization tolerance. The choice of estimator  determines the relevant variance prefactors $w_s$ in Eq.~\eqref{eq:partial_derivative_variance}, which in turn dictate the shot scaling prefactors in Eq.~\eqref{eq:shots_per_case}. Throughout this subsection, we use $d_T$ to denote a \emph{target steady-state relative cost gap}. The goal is to allocate a shot budget large enough such that, in expectation, gradient-estimation noise does not prevent the optimizer from reaching the target $d_T$ relative to the statistical-noise-free steady-state level; that is, the steady-state cost attained with state-vector simulations, corresponding to the limit of an infinite number of shots.
 In practice, for a fixed depth $p$, Thm.~\ref{thm:shots_per_iteration_sgd} together with Eq.~\eqref{eq:shots_per_case} suggest two regimes: for finite-difference gradient estimation, the required number of shots per cost evaluation decreases approximately as $1/m$, whereas for parameter-shift it remains approximately constant, up to graph structure and depth-dependent prefactors. We therefore recommend conducting a short calibration phase on small representative instances of the target graph family. The inverse-size extrapolation rule below applies to the finite-difference case:
\begin{enumerate}
    \item \textbf{Fix the target and the estimator.}
    Select the circuit depth $p$, the target tolerance $d_T$, and the gradient estimation method (finite-difference or parameter-shift).
    \item \textbf{Calibrate on small representative graphs.}
    Sample a small set of representative instances $\{G_{m_0}^{j}\}_{j=1}^J$ from the target graph family at a modest edge count $m_0$. For each instance, first run the same depth-$p$ optimization in a shot-noise-free mode, replacing sampled cost and gradient evaluations by state-vector values, and record the resulting steady-state reference cost $F_p(\vtheta^{\,*};G_{m_0}^{j})$. Then rerun the optimizer under several candidate shot budgets per cost evaluation, $S\in\{S_1,S_2,\dots\}$. For each $S$, record the value of the empirical relative cost gap from the steady-state reference cost
    \[
    d_t \;=\; \frac{F_p(\vtheta^{\,t};G_{m_0}^{j})-F_p(\vtheta^{\,*};G_{m_0}^{j})}{\bigl\lvert F_p(\vtheta^{\,*};G_{m_0}^{j})\bigr\rvert},
    \]
    estimated from the final portion of the run (e.g., the last $T$ iterations after convergence diagnostics indicate a plateau). Then choose, for each graph, the smallest shot count $S^{j}$ such that the observed steady-state relative cost gap is within the desired $d_T$. Finally, set the calibrated shot budget as a robust aggregate over the calibration set (e.g., the median or mean),
    \[
    S_{m_0} \;:=\; \mathrm{median}_{j \in J}\, S^{j} \qquad (\text{or mean}).
    \]
    \item \textbf{Fix transferable angles.}
If the graph family is \BS convergent and admits angle transfer across sizes, the calibration procedure simplifies: first, optimize the angles once on the calibration set to obtain transferable angles  $\vec{\theta}$, which can then be reused on larger graphs, following \cite{brandao2018fixed,wurtz2021fixedangle,wybo2025missing}. Then, determine the minimal shot count $S_{m_0}$ for which the cost estimate at $\vec{\theta}$ achieves the desired accuracy, relative to  the statistically noise-free simulation, within the target confidence level.
    \item \textbf{Scale the shots to larger instances.}
    For a larger instance with $m$ edges, choose the shots per cost evaluation according to the scaling rule suggested by the analytical bounds. For the finite-difference case considered here, this gives
    \begin{equation}
        \label{eq:shots_scaling_recipe}
        S_m \;\approx\; S_{m_0}\,\frac{m_0}{m},
    \end{equation}
    according to Eq.~\eqref{eq:shots_per_case}.
    \item \textbf{Run and verify.}
    Run QAOA on the larger instances using either the transferred angles or re-optimized angles if transfer is unreliable for the graph family under study.
\end{enumerate}  
These steps are intentionally conservative. The analytical bounds are sufficient conditions rather than tight prescriptions, and they rely on the assumptions stated in Secs.~\ref{sec:analytical_result_b}--\ref{sec:analytical_result_c}. Still, they make explicit how the shot budget depends on the problem size, the gradient estimation method, and the circuit depth through concrete prefactors. In the numerical results section below, we apply this \textit{\textbf{calibration-and-scaling}} methodology and demonstrate it numerically.
\section{Experimental results} \label{sec:exp_results}
In what follows, we present numerical results to demonstrate and support our analytical analysis.

\subsection{Computational setup}
All experiments were conducted with noiseless classical simulations, with finite shot count, using the PennyLane Python package \cite{pennylane,pennylane2,pennylane3}.
In all cases, we initialized the tunable \gammabeta parameters randomly, unless stated otherwise. We chose to work with the following sets of graphs, where the same instances are used repeatedly across the various experiments:
\begin{enumerate}
    \item Randomly generated 3-regular graphs with $N=\{10,12,\ldots,22\}$ nodes, corresponding to $m\in\{15,18,21,24,27,30,33\}$ edges (10 instances per size).
    \item Sparse \ER graphs with $N\in\{10,12,\ldots,22\}$ nodes with edge probability $p=\frac{3\pm0.1}{N-1}$. 
    These sparse \ER graphs have an expected degree of $v=3\pm0.1$, making them comparable to the 3-regular graphs with $m\in\{15,18,21,24,27,30,33\}$ edges (10 instances per size).
    \item Connected random graph structures with $m\in\{15,18,21,24,27,30,33\}$ edges (10 instances per size). Each instance was generated by first sampling a feasible number of vertices $N\le26$, constructing a random tree to ensure connectivity, and then adding randomly selected missing edges until the graph had exactly $m$ edges.
    \item 
    2-regular graphs, i.e.\ cycles, for which there is full isomorphism of all edge $p$-neighborhoods, thus providing a clear numerical demonstration of our cost gap analysis shown in Appendix~\ref{app:parameter_shift_rule_results}).
\end{enumerate} 
\subsection{Results supporting Result A - shot-number scaling for estimating the QAOA cost expectation (Thm.~\ref{thm:shots_per_iteration})}
Here we illustrate the main prediction of Sec.~\ref{sec:analytical_result_a}: for a fixed circuit depth and comparable sparse graph families, the number of shots needed to estimate the QAOA cost to a prescribed relative accuracy decreases as the graph size grows. 
We begin with Fig.~\ref{fig:ecost_variance_erdos_renyi}, which examines the scaling of $\operatorname{Var}(C_p)$ because this is the quantity entering the common-practice shot estimate in Eq.~\eqref{eq:number_of_shots_common_practice}. Thus, the first numerical check concerns the variance-based route to the shot-number scaling. The Janson-bound route is separate: it does not use $\operatorname{Var}(C_p)$, but rather the bounded dependency degree of the edge-measurement variables, and is later probed more directly by the shot-count experiment in Fig.~\ref{fig:qaoa_measure_concentration_grid}.
Figs.~\ref{fig:regular} and~\ref{fig:er}
report the variance $\operatorname{Var}(C_p)=\langle C_p^2\rangle-\langle C_p\rangle^2$ computed using state-vector simulations on random 3-regular and sparse \ER graphs, respectively. 
For each graph size, we averaged over 10 graph instances and evaluated the variance for both $p=1$ (solid blue) and $p=2$ (solid red) at randomly sampled angles $(\gamma,\beta)$, without running an optimization loop.
    \begin{figure}[H]
        \centering
        \subfloat[3-Regular]
        {\includegraphics[width=0.45\linewidth]{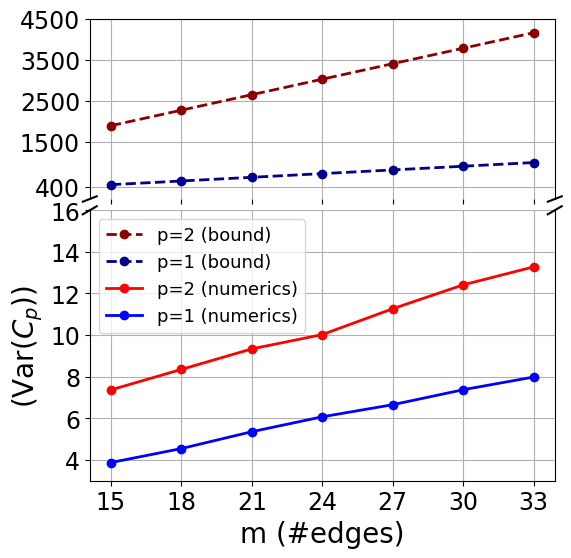}\label{fig:regular}}
        \subfloat[\ER]
        {\includegraphics[width=0.45\linewidth]{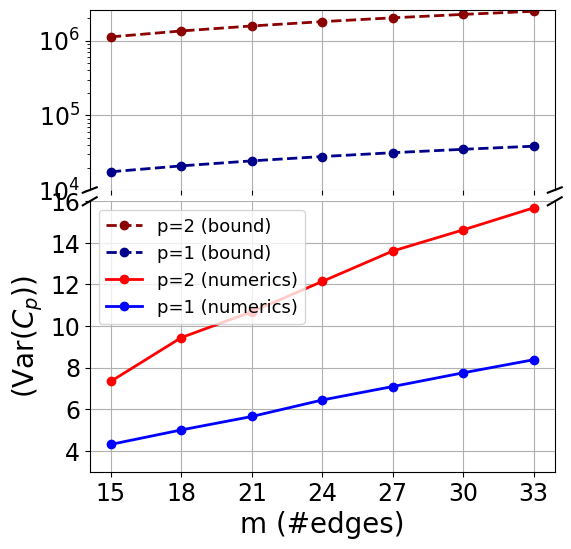} \label{fig:er}}
        \caption{\protect{\textbf{Scaling of the QAOA cost variance with graph size.} 
        Dashed curves show the analytical upper bound from Eq.~\eqref{eq:variance_upper_bound}, while solid curves show the state-vector evaluated variance $\operatorname{Var}(C_p)$ for both $p=1$ (blue) and $p=2$ (red). \textbf{(a)} Random 3-regular graphs with $m=\frac{3N}{2}$ edges. \textbf{(b)} Sparse \ER graphs with edge probability $3/(N-1)$, which yields an average of $m=\frac{3N}{2}$ edges. For each size, the values are averaged over 10 graph instances. It is seen that the variance grows approximately linearly with $m$, and the analytical bound remains conservative.}}
        \label{fig:ecost_variance_erdos_renyi}
    \end{figure}
    The near-linear growth in both panels of Fig.~\ref{fig:ecost_variance_erdos_renyi} is the main numerical point: over the tested range, the variance of the cost  scales extensively with $m$. The dashed analytical upper bounds, calculated according to Eq.~\eqref{eq:variance_upper_bound}, derived in Ref.~\cite{farhi2014quantum}, which by definition scale linearly with $m$, are visibly non-tight.
    This gap is expected, because the bound in  Eq.~\eqref{eq:variance_upper_bound} was obtained   using the maximal possible covariance of $1$ (see Ref.~\cite{farhi2014quantum}). Numerically, these covariance contributions are typically much smaller and can even be negative, thereby reducing the total variance of the edge sum relative to the worst-case bound.
    The analytical results of Sec.~\ref{sec:analytical_result_a} show an inverse relationship between the number of edges in the graph and the number of shots needed to achieve the same relative error $\delta$ and confidence level $1-\epsilon$ for the cost estimation.     
    Fig.~\ref{fig:qaoa_measure_concentration_grid} next shows that larger graphs of similar structure require fewer shots to estimate the cost to the same relative error $\delta$, in accordance with Observation~\ref{thm:shots_based_on_common_practice} and Thm.~\ref{thm:shots_per_iteration}. Moreover, it shows that in practice, the ``structural similarity" condition may be relaxed and that the same observation holds qualitatively for the generated random connected graph structures.
    The experiment was performed using noiseless QAOA MaxCut simulations on the sets of 3-regular, sparse \ERc and random connected graphs with a fixed depth of $p=1$ and increasing graph size.
    We performed the experiment with $\delta=0.2$ and random circuit parameters $(\gamma,\beta)$. 
    We used an exhaustive search to find the minimal number of shots $n$ needed to achieve $\epsilon$ values in  $\left[0.2,0.15,0.1,0.05,0.025\right]$ where $\epsilon$ is the probability that the estimation error exceeds $\delta \costexpect$, as in Thm.~\ref{thm:shots_per_iteration}. 
    Note that as $\epsilon$ grows, the number of required measurements decreases, as expected, and that the same qualitative behavior is observed throughout for 3-regular, \ERc and random connected graphs.  
    Fig.~\ref{fig:qaoa_measure_concentration_grid} therefore provides a direct numerical illustration of Result~A: provided the cost remains extensive, larger graphs (of similar local structure) require fewer measurements to reach the same relative estimation accuracy.    
    \begin{figure}[H]
    \centering
    \subfloat[3-Regular]
    {\includegraphics[width=0.33\textwidth]{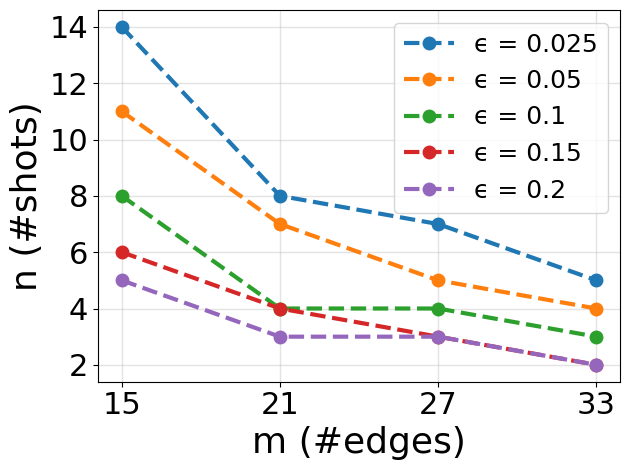}
    }
    \subfloat[\ER]{\includegraphics[width=0.33\textwidth]{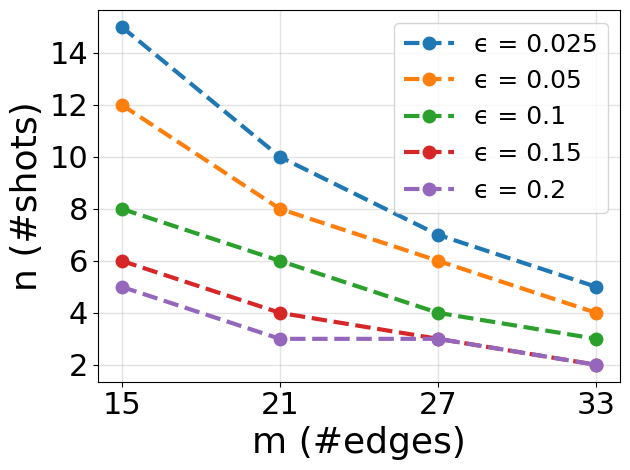}}
    \subfloat[Random Structure]{\includegraphics[width=0.33\textwidth]{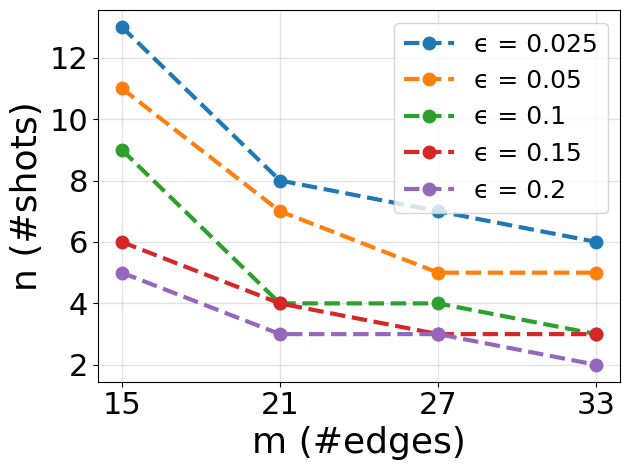}}
         \caption[Required shots for relative-error cost estimation]{\textbf{Required shots for relative-error cost estimation versus graph size.} The figure shows the smallest number of shots required to estimate $\langle \hat{C}_p\rangle$ to within a relative error $\delta=0.2$ with confidence level $1-\epsilon$, as a function of the number of edges $m$ at a fixed depth of $p=1$. One curve is plotted for each value of $\epsilon$, and the same sampled circuit angles are used across all graph sizes. \textbf{(a)} 3-regular graphs. \textbf{(b)} Sparse \ER graphs with edge probability $3/(N-1)$. \textbf{(c)} Connected simple graphs with specified number of edges. For each size, we use a single random graph instance from the corresponding family. In all cases, the required shot count decreases with graph size, in qualitative agreement with Eq.~\eqref{eq:inverse_relation_shots_m}.}
        \label{fig:qaoa_measure_concentration_grid}
    \end{figure}     
    Finally, Fig.~\ref{fig:numerical_illustration} provides a concrete numerical demonstration of the general trend established above. 
    To that end, we consider randomly sampled 3-regular graphs of increasing size, using 10 graph instances per size. 
    Building on prior tree-QAOA results \cite{brandao2018fixed,streif2019treeqaoa,wurtz2021fixedangle,wybo2025missing}, we demonstrate that a fixed circuit depth $p$  and fixed angles \gammabeta maintain a narrow range of approximation ratios, despite using a shot count that decreases with the number of edges $m$.
    Specifically, we optimize the angles \gammabeta separately for the 10 instances of the smallest graph size, and then apply their circular  mean to all larger graph sizes without further outer-loop optimization. The reported approximation ratios $r_p$ are averaged over the instances of each size.
    This observed reduction in the required sampling budget serves as a qualitative numerical validation of Eq.~\eqref{eq:inverse_relation_shots_m}: for graph instances with similar local structure, the number of shots in QAOA MaxCut needed to maintain a target relative estimation error  typically scales inversely with graph size. 
    By combining this reduced sampling requirement with the parameter concentration of tree-QAOA, we moreover maintain stable performance in terms of the approximation ratio.
    \begin{figure}[H]
        \centering
        \includegraphics[width=0.55\linewidth]{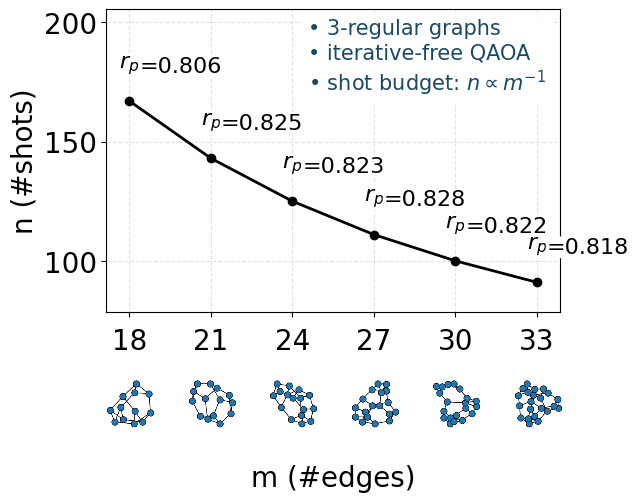}
        \caption{\textbf{Numerical illustration of inverse shot scaling on random 3-regular graphs.} The figure shows the shot budget, $\shots$, required to  estimate the QAOA cost, 
        as a function of the number of edges $m$ for a family of random 3-regular graph instances (10 graph instances per size). The evaluation uses fixed angles at depth $p=2$ optimized solely for the smallest graph and applied throughout across all larger instances. The numerical labels indicate the achieved approximation ratios $r_p$. 
        Notably, the required shot budget decreases with graph size while the approximation ratios remain stable.}
        \label{fig:numerical_illustration}               
    \end{figure}

\subsection{Results supporting Result B - shot-number scaling along a stochastic gradient descent QAOA optimization (Thm.~\ref{thm:shots_per_iteration_sgd})}

This section provides numerical support for the analytical result presented in Sec.~\ref{sec:analytical_result_b} regarding the shot count required in SGD QAOA optimization. The goal is to demonstrate the scaling ingredients comprising  Thm.~\ref{thm:shots_per_iteration_sgd}.
All optimization experiments reported here use finite-difference gradient estimates, so that the numerical comparison is aimed specifically at the finite-difference branch of Eq.~\eqref{eq:shots_per_case}. The case of the parameter-shift rule is addressed in Ref.~\ref{app:parameter_shift_rule_results}.

\par\medskip\noindent\textbf{Gradient estimator variance}\par

    We begin by examining how the variance of the 
    gradient estimators  scales with graph size. This is the natural optimization-side analogue of the cost-variance study above, since the shot budget required by SGD depends on the noise level of the estimated gradient. In particular, it tests the variance relation summarized in Eq.~\eqref{eq:partial_derivative_variance}. To quantify the gradient noise entering SGD, we evaluate the depth-$1$ 
    finite-difference and gate-wise parameter-shift estimators on random 3-regular graphs of sizes $N\in\{10,12,\ldots,22\}$ (10 instances per graph size).
    To that end, we used a single, fixed, randomly sampled value per parameter $(\gamma_1,\beta_1)$. 

    Fig.~\ref{fig:partial_derivative_variance} shows the estimator-dependent trends predicted by Eq.~\eqref{eq:var_grad_bound_scaling_cases}: the finite-difference estimator variances with respect to $\gamma_1$ and $\beta_1$ increase approximately linearly with the number of edges $m$, while the parameter-shift estimator variances grow approximately quadratically. 
    For the finite-difference branch, this linear variance growth, combined with the extensive cost scaling in Assumption~\ref{ass:qaoa_cost_lower_bound} and the linear scaling $\mu=\Theta(m)$ assumed in Eq.~\eqref{eq:shots_per_case}, supports the conclusion that the relative stochastic error, $\xi$, decreases with graph size for a fixed shot budget (see Eq.~\eqref{eq:relative_stochastic_error_final}). Equivalently, maintaining a fixed target relative error, $\xi^*$, requires a number of shots per cost evaluation that decreases as $1/m$ in the finite-difference setting (see Eq.~\eqref{eq:sufficient_shots_per_cost_evaluation}).
    
    \begin{figure}[H]
        \centering
        \subfloat[Finite-difference]
        {\includegraphics[width=0.5\textwidth]{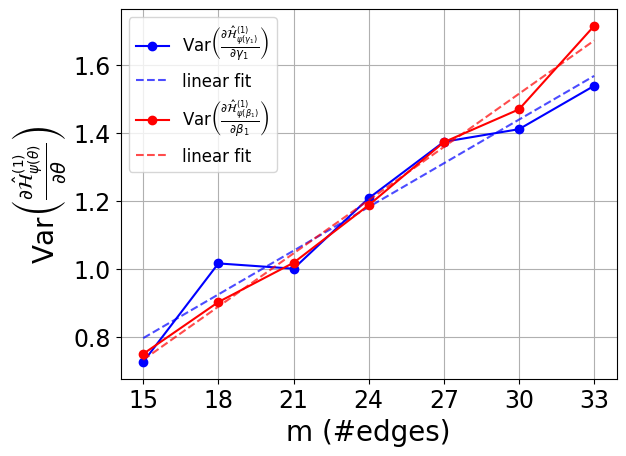}}
        \subfloat[Parameter-shift rule]
        {\includegraphics[width=0.5\linewidth]{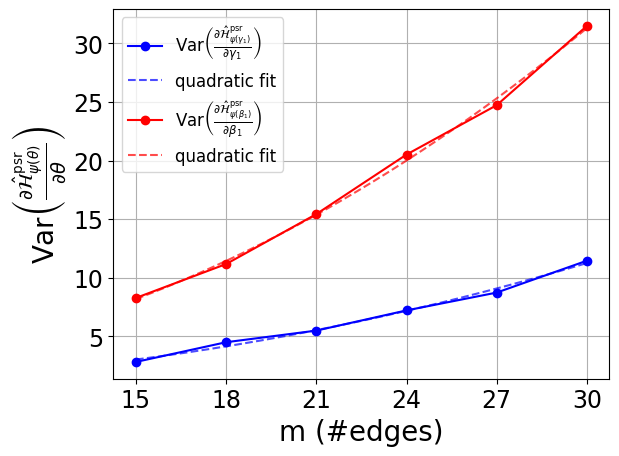}}
        \caption[Gradient-estimator variance versus graph size]{\protect{
        \textbf{Sample variance of depth-$1$ gradient estimators on random $3$-regular graphs, comparing (a) the \textit{finite-difference approximation} (left) and (b) the \textit{parameter-shift estimator} (right).} Solid blue and red lines denote the estimators for $\partial C_p/\partial \gamma_1$ and $\partial C_p/\partial \beta_1$, respectively. 
        A fixed number of shots $n_p$ per cost evaluation is used in each experiment: $n_p=1000$ in the finite-difference experiment, and $n_p=10$ in the parameter-shift rule experiment. 
        For each graph size $N\in\{10,12,\ldots,22\}$, the reported variance is averaged over $10$ graph instances and estimated from $100$ gradient samples per instance. 
        A single, fixed evaluation point $(\gamma_1,\beta_1)$, chosen at random, is applied across all instances.
        Dashed curves indicate the corresponding scaling fits: the finite-difference variance grows approximately linearly with the number of edges $m$, while the parameter-shift estimator variance grows approximately quadratically with $m$,
        in accordance with Eq.~\eqref{eq:var_grad_bound_scaling_cases}.}
        }
    \label{fig:partial_derivative_variance}
    \end{figure}
    \par\medskip\noindent\textbf{Steady-state behavior under inverse shot scaling}\par
    We next study how the shot budget per cost evaluation can be reduced as the graph size increases while maintaining comparable optimization quality. As a summary statistic, we use the relative cost gap $d_t$, defined in Eq.~\eqref{eq:d_t_relative_cost_gap}.      
    For each graph size, we ran QAOA with SGD for $T=100$ iterations, treating the first $T_0=40$ iterations as the convergence phase (we verified numerically that after $T_0$, the mean remaining improvement in the relative cost expectation, averaged over all graph instances and optimization runs, dropped below $10^{-2}$), 
    and summarized steady-state performance using the mean value over the remaining iterations $\{d_t\}_{t=T_0+1}^{T}$. In the experiments below, the number of shots per cost evaluation in the finite-difference gradient estimator is scaled explicitly as $\shots \propto 1/m$, in accordance with Eq.~\eqref{eq:shots_per_case}. We fixed the QAOA depth to $p=1$ and applied shot noise only to the gradient estimator, while the objective $F(\theta^t)$ used to compute the relative cost gap $d_t$ was computed as an exact state-vector expectation. 
    Fig.~\ref{fig:3regular_sgd_relative_cost_gap} shows that using $\shots \propto 1/m$ preserves comparable steady-state optimization accuracy across the tested sizes. This is consistent with the finite-difference branch in Eq.~\eqref{eq:shots_per_case} and provides direct numerical support for Result B in the studied regime. The same qualitative behavior is observed for both higher shot schedules (blue)  and lower shot schedules (black), leading to correspondingly smaller and larger relative cost gaps, as expected.
    \begin{figure}[H]
        \centering
        \subfloat[\textbf{3-Regular graphs:} shots per cost evaluation, set according to Eq.~\eqref{eq:shots_scaling_recipe}.]
{\includegraphics[width=0.45\linewidth]{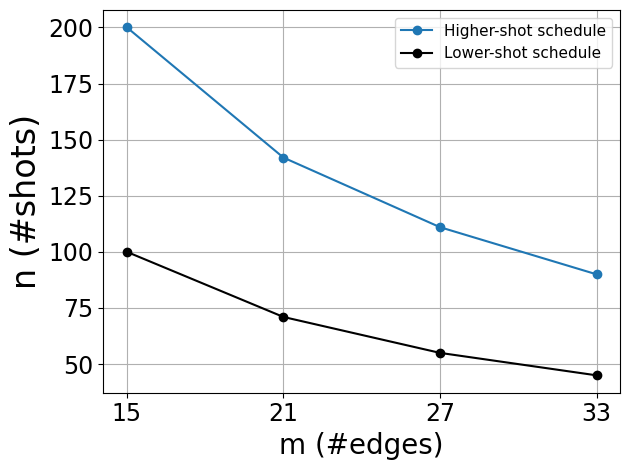}}
        \quad
       \subfloat[Corresponding steady-state relative cost gap $d_t$ versus number of edges $m$.]
{\includegraphics[width=0.45\linewidth]{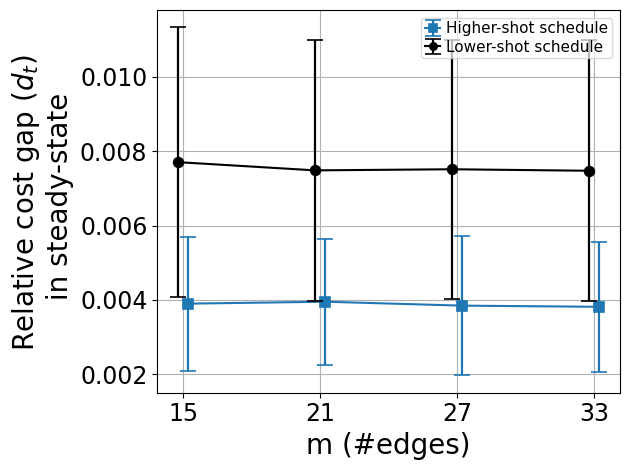}}
        \caption[QAOA-SGD under inverse shot scaling on 3-regular graphs]{\protect\textbf{Numerical performance of QAOA-SGD with shots scaled inversely with graph size on random 3-regular graphs with $N\in\{10,14,18,22\}$ nodes.} We used $n=3000/m$ (blue) and $n=1500/m$ (black) shots at a fixed circuit depth $p=1$. 
        \textbf{(a)} 
        Shots per cost evaluation used to estimate the gradients
        as a function of the number of edge. \textbf{(b)} 
        The relative cost gap ($d_t$) in steady-state as a function of edges count. 
        The plotted steady-state value is the mean of $d_t$ over the last $T-T_0$ iterations, with $T=100$ and $T_0=40$. 
        The observed steady-state gap varies only weakly across the tested sizes.
        }
\label{fig:3regular_sgd_relative_cost_gap}
    \end{figure}
    \subsection{Results supporting Result C - Iteration scaling in QAOA with stochastic gradient descent (Thm.~\ref{thm:req_iteration_number})} 
    We now turn to the scaling of the number of optimization iterations. 
    The following results illustrate that the convergence rate remains consistent across sizes, while the shot budget required for stable stochastic gradient descent updates decreases with graph size.
    Our focus in this subsection is the optimization trajectory and the resulting size-robust convergence pattern, which supports the scaling claim underlying Thm.~\ref{thm:req_iteration_number}.
    \begin{figure}[H]
        \centering
        \subfloat[
        \textbf{Random 3-regular graphs:} the relative cost gap $d_t$ (defined in Eq.~\eqref{eq:d_t_relative_cost_gap}) as a function of SGD iterations with \textbf{finite-difference approximation} and $\shots \propto 1/m$.
        ]
{\includegraphics[width=0.45\linewidth]{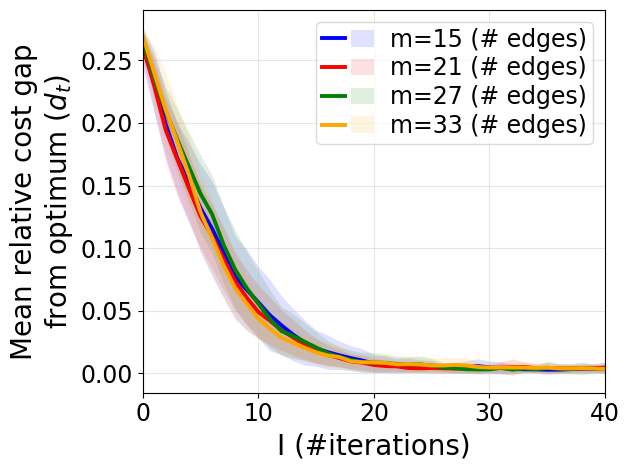}}
        \quad
       \subfloat[Higher resolution: the corresponding standard deviation $\sigma(d_t)$ as a function of iterations.]
       {\includegraphics[width=0.45\linewidth]{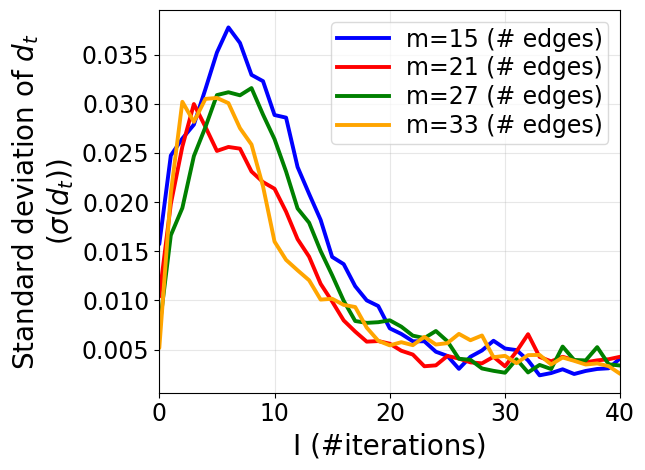}}
        \caption[Representative SGD trajectories under inverse shot scaling]{\protect\textbf{Relative cost gap ($d_t$) trajectories on random 3-regular graphs under inverse shot scaling using the finite-difference approximation.} 
        Graphs with $N\in\{10,14,18,22\}$ nodes, corresponding to $m\in\{15,21,27,33\}$ edges, were considered, with 10 graph instances per size and 5 independent finite-shot SGD trajectories per instance, at circuit depth $p=1$. All trajectories were initialized from the same fixed parameter vector $(\gamma_1,\beta_1)=(0,0.5)$. The values $d_t$ were evaluated from exact state-vector cost expectations along the SGD trajectories, using the corresponding noiseless reference value according to Eq.~\eqref{eq:d_t_relative_cost_gap}.
        A finite-difference gradient estimator was employed, using a number of shots per circuit evaluation proportional to $1/m$.  
        \textbf{(a)} The mean relative cost gap $d_t$ is shown as a function of the iteration number (solid lines) together with the  standard deviation $\pm\,\sigma(d_t)$ (shaded). \textbf{(b)} The corresponding standard deviation $\sigma(d_t)$ isolated for clarity. 
        Across all the tested sizes, the convergence rates and trajectory spreads remain comparable while the shot budget per cost evaluation decreases with increasing graph size, illustrating Eq.~\eqref{eq:shots_per_case} in the finite-difference case.}
        \label{fig:3regular_sgd_finite_difference}
    \end{figure}
    To complement the steady-state analysis, Fig.~\ref{fig:3regular_sgd_finite_difference} shows full SGD trajectories on random 3-regular graphs 
    using a finite-difference gradient estimator with a shot budget that decreases inversely with graph size. 
    Here we use random 3-regular graphs because their local structure is comparable across sizes, which makes the cross-size comparison consistent. The stochasticity enters only through the sampled gradient, whereas the cost values used for the relative gap ($d_t$) calculations rely on exact state-vector expectations. Notably, the convergence curves  and their spread remain comparable across the tested sizes while the shot budget per cost evaluation decreases with the number of edges according to $\shots \propto 1/m$. This size-independent convergence behavior is consistent with the constant iteration conclusion of Thm.~\ref{thm:req_iteration_number}. For the analogous trajectories obtained with the  parameter-shift estimator, see Appendix~\ref{app:parameter_shift_rule_results}.
    Taken together, the numerical results of this section support the analytical picture developed in Secs.~\ref{sec:analytical_result_a}-\ref{sec:analytical_result_c}: for graph families with comparable local structure and fixed QAOA depth, the relevant variances grow extensively, 
    which in turn enables the shot budget needed per cost evaluation to decrease with graph size when using the finite-difference approximation, while the number of SGD optimization iterations remains independent of the system size.
    
\section{Discussion}
\label{sec:conclusions}
By analyzing the statistical behavior of the MaxCut cost function and its gradients, we established rigorous bounds on the sampling complexity and iteration count required for a target relative performance level. Our analysis shows that for large instances with fixed maximum degree $v$ and circuit depth $p$, the number of shots required to estimate the QAOA cost function to within a relative performance typically scales inversely with the number of edges $m$. These "blessings of dimensionality" results extend earlier findings regarding the transferability of optimal parameters from small to large instances in regular graphs \cite{brandao2018fixed,streif2019treeqaoa,wurtz2021fixedangle,wybo2025missing}. We show that for such instances, not only is the classical optimization loop potentially redundant, but the required shot budget actually decreases as the system size grows.
Extending our analysis to Stochastic Gradient Descent (SGD), we examined the sampling and iteration requirements for achieving a specific convergence target. We found that for graphs with a consistent local structure as $m$ grows, the shot budget for the parameter-shift rule is independent of system size. In contrast, using finite differencing allows the shot count to decrease as $1/m$, while the total iteration count remains constant. Together, these results imply that the overall time-to-solution (TTS) for QAOA may actually improve as the problem instance grows, provided a relative performance metric is used.
A decreasing time-to-solution with problem size is a unique phenomenon within the standard algorithmic landscape, both classical and quantum, where larger problem instances typically demand increased computational resources, such as more measurement shots or optimization steps. Our analysis provides deeper insight into the interplay between local graph structure and the computational resources required when performing QAOA for MaxCut, demonstrating how local concentration can be leveraged to reduce sampling overhead. Ultimately, these findings provide both analytical and numerical evidence for the scalability of QAOA in solving large-scale MaxCut instances.
The methodology and results presented here can be further investigated to find similar behaviors in other combinatorial problems solved with QAOA that exhibit similar concepts of causality and local dependence. Extending these findings to other variational quantum algorithms (VQAs) where relative performance is the primary metric would also be of significant interest. Note, however, that in cases where absolute accuracy is required, such as in the variational quantum eigensolver (VQE), where a fixed energy error is desired irrespective of molecular size, such an advantage is not expected. It would also be of interest to provide tighter bounds on the statistical behavior of the measurement operator by formalizing the variance expression as a function of subgraph types.
Finally, we leave open the question posed in Conjecture~\ref{conj:open_ended_edge_positive}: whether every open-ended radius-$p$ edge neighborhood has a positive local contribution for any parameters assignment.
\newpage
\bmhead{Acknowledgments}
We thank Prof. Reuven Cohen for useful discussion. 
We acknowledge the use of the open-source PennyLane software for implementing and simulating the quantum circuits used in this work \cite{pennylane}.

\bmhead{Competing Interests}
The Authors declare no competing financial or non-financial interests.

\bmhead{Data Availability}
All data generated or analyzed during this study are included in this published article and its supplementary information files.

\bmhead{Code Availability}
The underlying code for this study is available in GitHub and can be accessed via this link \url{https://github.com/inbarchef/qaoa-shot-analysis}.

\bmhead{Author Contributions} 
US and AM supervised the project. All authors developed the theory. IC designed the computational experiments and performed all computations. IC and AM wrote the manuscript with support from US. All authors helped shape the research, analysis, and manuscript.
\begin{appendices}

\section{Cost function scaling assumption}\label{app:cost_function_scaling}
\noindent
  
In what follows, we provide reasoning behind Assumption~\ref{ass:qaoa_cost_lower_bound}, restated next for completeness:
\textit{Let $G$ be a graph with $m$ edges, and let $\vec{\theta}\in\mathbb{R}^{2p}$ denote the QAOA tunable parameters. Then there exists a constant $\kappa>0$, independent of $m$, such that}
\begin{equation}
\costexpect \equiv \big\lvert \FpThetaG \big\rvert \;\ge\; \kappa\, m .
\end{equation}
Below we provide three complementary arguments supporting this extensivity. Reasoning I fixes the graph $G$ and averages only over the optimization randomness. Reasoning II fixes the parameters $\vec{\theta}$ and averages over the graph ensemble, including random regular graphs and sparse \ER graphs. Reasoning III supplies the local positivity criterion needed in Reasoning II, by identifying edge-neighborhood types with strictly positive contribution.
\paragraph{Reasoning I: Initialization of $\vec{0}$}
Let $(\boldsymbol{\vec\gamma}^{0},\boldsymbol{\vec\beta}^{0})=(\boldsymbol{0},\boldsymbol{0})$ be the initial parameters for the optimization. 
At this point the QAOA unitary is the identity, so the state is $\ket{\psi_0}=\ket{+}^{\otimes q}$ (with $q$ qubits/vertices). For any edge type $g$ and any depth $p$, the corresponding per-edge contribution evaluates to
\[
f_{p,g}(\boldsymbol{0},\boldsymbol{0})
=
\operatorname{Tr}\!\left[\ket{\psi_0}\bra{\psi_0} C_{j,k}\right]
=\frac{1}{2},
\]
since $\bra{+}^{\otimes q}Z_jZ_k\ket{+}^{\otimes q}=0$ (see also Eq.~\ref{eq:qaoa_problem_hamiltonian}). Hence, for every graph $G$ with $m$ edges,
\begin{equation}
F_p
(G,\boldsymbol{0},\boldsymbol{0})
=
-\sum_{g} w_g(G) f_{p,g}(\boldsymbol{0},\boldsymbol{0})
=
-\frac{1}{2}\sum_{g} w_g(G)
=
-\frac{m}{2},
\label{eq:initial_value_minus_m_over_2}
\end{equation}
using $\sum_g w_g(G)=m$.  
By assuming monotonicity in expectation along the iterative process, we get that for any time step $t$ it holds that
\begin{equation}
\label{eq:qaoa_cost_m_div_2}
\mathbb{E}_{\hat{g}^t}\left[\,\bigl\vert\langle C_p\rangle\bigr\vert\,\right]=-\mathbb{E}_{\hat{g}^t}\!\left[F_p(\boldsymbol{\gamma}^{t},\boldsymbol{\beta}^{t})\right]\;\ge\;
-\;F_p(\gamma^0,\beta^0) =
\;
-\;F_p(\boldsymbol{0},\boldsymbol{0})=\frac{m}{2},
\end{equation}
corresponding to $\kappa=1/2$.
\paragraph{Reasoning II: Benjamini--Schramm local limit}
Following the review of \BS convergence in Sec.~\ref{subsec:benjamini_schramm_convergence}, here we elaborate on the edge-rooted formulation required below.
In particular, for every fixed radius $p$ and every rooted edge neighborhood type $g_0$, edge rooted \BS convergence implies that the finite-size probabilities converge as $m\to\infty$ \cite{benjamini_schramm_2001}.
Let $G_m$ denote a random graph instance drawn from a graph ensemble with $m$ edges, let $\mathbb{E}_{G}$ denote expectation over its law, and let $e_m$ be a uniformly sampled random edge of $G_m$. We write $B_p(e_m;G_m)$ for the \emph{radius-$p$ edge-neighborhood} of that sampled edge, namely the rooted subgraph induced by all vertices at graph distance at most $p$ from either endpoint.
Then for any fixed $p$ and rooted edge-neighborhood type $g_0$, we define the exact finite-size probability by
\begin{equation}\label{eq:benjamini_schram_probability_convergence}
\rho_m(g_0):=\Pr\bigl(B_p(e_m;G_m)\cong g_0\bigr)\in[0,1],
    \qquad
    \rho(g_0):=\lim_{m\to\infty}\rho_m(g_0).
\end{equation}
Thus, $\rho_m(g_0)$ is the exact finite-size probability that a uniformly sampled edge has a rooted neighborhood of type $g_0$, and edge-rooted \BS convergence means precisely that $\rho_m(g_0)\to \rho(g_0)$ as $m\to\infty$. In the present reasoning, we choose $g_0$ so that $\rho(g_0)>0$. \\ 
Using the decomposition of the QAOA cost into rooted neighborhood types, the cost on a graph realization $G_m$ can be written as
\[
\FpThetaGm
=
-\sum_{g} w_g(G_m)\, f_{p,g}(\vec{\theta}),
\]
where $w_g(G_m)\ge 0$ denotes the number of edges in $G_m$ whose radius-$p$ neighborhood has type $g$, and $0\le f_{p,g}(\vec{\theta})\le 1$ for all $g$.
Fix parameters $\vec{\theta}$ and assume that for some distinguished type $g_0$ one has $f_{p,g_0}(\vec{\theta})=c>0$, with $c$ independent of $m$, then
\[
\FpThetaGm
\le
-w_{g_0}(G_m)\,f_{p,g_0}(\vec{\theta})
=
-c\,w_{g_0}(G_m).
\]
For a fixed realization of $G_m$, since $e_m$ is sampled uniformly from the $m$ edges,
\[
\Pr\bigl(B_p(e_m;G_m)\cong g_0 \mid G_m\bigr)=\frac{w_{g_0}(G_m)}{m}.
\]
Taking expectation over the graph ensemble gives the exact finite-size identity
\[
\mathbb{E}_{G}\!\left[w_{g_0}(G_m)\right]
=
m\,\rho_m(g_0).
\]
Averaging the bound $\FpThetaGm \le -c\,w_{g_0}(G_m)$ over the graph ensemble gives
\[
\mathbb{E}_{G}\!\left[\abs{\FpThetaGm}\right]
=
-\mathbb{E}_{G}\!\left[\FpThetaGm\right]
\ge
c\,\mathbb{E}_{G}\!\left[w_{g_0}(G_m)\right]
=
c\,\rho_m(g_0)\,m.
\]
This lower bound supports the same extensive scaling as above, but through a different averaging mechanism: in Reasoning~I the expectation is over the optimization randomness for a fixed graph, whereas here it is over the graph ensemble for fixed parameters.
Since $\rho_m(g_0)\to \rho(g_0)>0$, there exists $m_0$ such that $\rho_m(g_0)\ge \rho(g_0)/2$ for all $m\ge m_0$. Therefore, for all sufficiently large $m$,
\[
\mathbb{E}_{G}\!\left[\abs{\FpThetaGm}\right] \ge \frac{c\,\rho(g_0)}{2}\,m.
\]
Moreover, if one can identify a rooted edge-neighborhood type $g_0$ and a constant $\rho_\star>0$ such that $\rho_m(g_0)\ge \rho_\star$ for all $m$ under consideration, then the same computation gives the finite-size bound $\mathbb{E}_{G}\!\left[\abs{\FpThetaGm}\right]\ge c\,\rho_\star\,m$ for every such $m$. 
Reasoning II thus provides an average-case asymptotic support for the same extensive scaling on graph families such as random $v$-regular graphs and sparse Erd\H{o}s--R\'enyi graphs, whose edge-rooted local distributions converge as reviewed in Sec.~\ref{subsec:benjamini_schramm_convergence}; under the stronger condition $\rho_m(g_0)\ge \rho_\star>0$, it gives an $m$-independent lower bound for every $m$ in expectation over the graph ensemble. 
For example, let $g_0$ be the $v$-regular tree-like radius-$p$ neighborhood of an edge. If a finite $v$-regular graph $G$ has girth larger than $2p+1$, then every edge of $G$ has neighborhood type $g_0$. Hence $\rho_m(g_0)=1$, and the finite-size bound holds with $\rho_\star=1$.
\paragraph{Reasoning III: Edge contribution lower bound}
We now provide a complementary local argument for identifying edge-neighborhood types $g_0$ which  yield a strictly positive contribution $f_{p,g_0}(\vec{\theta})>0$. This completes the remaining input required for Reasoning II: once such a type appears with positive asymptotic density, the extensive lower bound follows from the decomposition in Eq.~\eqref{eq:qaoa_sum_over_subgraphs}. 
\begin{lemma}[Single-qubit expectation and distant two-qubit factorization]
\label{lem:far_vertex_decorrelation}
Let $G=(V,E)$ with $N=\vert V \vert$ nodes and $m=\vert E \vert$ edges, and let $\ket{\psi_p(\vec{\theta})}$ be the depth-$p$ QAOA state defined in Eq.~\eqref{eq:psi_gamma_beta}. For every vertex $r\in V$, the single-qubit $Z$ expectation vanishes:
\begin{equation}
\label{eq:single_Z_exp_vanishes}
\bra{\psi_p(\vec{\theta})}Z_r\ket{\psi_p(\vec{\theta})}=0. 
\end{equation}
Moreover, if two vertices $r,s\in V$ have a graph distance - defined as the shortest, or geodesic, path between them - larger than $2p$, then their two-qubit $Z$ correlator factorizes and vanishes:
\begin{equation} 
\label{eq:two_ZZ_exp_vanishes}
\bra{\psi_p(\vec{\theta})}Z_rZ_s\ket{\psi_p(\vec{\theta})}
=
\bra{\psi_p(\vec{\theta})}Z_r\ket{\psi_p(\vec{\theta})}
\bra{\psi_p(\vec{\theta})}Z_s\ket{\psi_p(\vec{\theta})}
=0.
\end{equation}
\end{lemma}
\begin{proof}
Consider the global bit-flip operator
\[
X_{\mathrm{all}}:=\prod_{i=1}^NX_i .
\]
The QAOA state is invariant under this global bit-flip symmetry. To see that, note that $X_{\mathrm{all}}$ commutes with the mixer Hamiltonian and with every MaxCut term in the cost Hamiltonian, since for every edge $(j,k)\in E$,
\[
X_{\mathrm{all}} Z_jZ_k X_{\mathrm{all}}=Z_jZ_k.
\]
Moreover, $X_{\mathrm{all}}\ket{+}^{\otimes N}=\ket{+}^{\otimes N}$, hence
\[
X_{\mathrm{all}}\ket{\psi_p(\vec{\theta})}
=
\ket{\psi_p(\vec{\theta})}.
\]
Using $X_{\mathrm{all}}Z_rX_{\mathrm{all}}=-Z_r$, we obtain
\[
\bra{\psi_p(\vec{\theta})}Z_r\ket{\psi_p(\vec{\theta})}
=
\bra{\psi_p(\vec{\theta})}
X_{\mathrm{all}}Z_rX_{\mathrm{all}}
\ket{\psi_p(\vec{\theta})}
=
-\bra{\psi_p(\vec{\theta})}Z_r\ket{\psi_p(\vec{\theta})},
\]
and therefore
\[
\bra{\psi_p(\vec{\theta})}Z_r\ket{\psi_p(\vec{\theta})}=0,\qquad\forall r \in V.
\]
If the graph distance between $r$ and $s$ is larger than $2p$, then the depth-$p$ reverse causal cones of $Z_r$ and $Z_s$ are disjoint. By the same causal-cone factorization argument used for disjoint edge costs in Eq.~\eqref{eq:uncorrelated_distinct_cost}, now applied to the single-qubit observables $Z_r$ and $Z_s$, the two-qubit expectation factorizes into the product of the single-qubit expectations above, and therefore vanishes.
\end{proof}
We next use this decorrelation property to show that the local cost contributions $f_{p,jk}$ along a long geodesic path cannot vanish simultaneously. Here, a 'zero contribution' corresponds to $f_{p,jk}(\vec{\theta})=0$, or equivalently, the vanishing of the expectation value $\bra{\psi_p(\vec{\theta})}C_{jk}\ket{\psi_p(\vec{\theta})} = 0$.
\begin{lemma}[A long geodesic path cannot be zero on every edge]
\label{lem:long_path_positive_edge}
Let
\[
P=(v_0,v_1,\ldots,v_\ell)
\]
be a geodesic path in $G$ with $\ell>2p$ edges. Then, for every fixed parameter vector $\vec{\theta}$, at least one edge in $P$ has strictly positive QAOA MaxCut contribution:
\[
\exists\, i\in\{1,\ldots,\ell\}
\quad\text{such that}\quad
f_{p,v_{i-1}v_i}(\vec{\theta})>0.
\]
\end{lemma}
\begin{proof}
Write $\ket{\psi_p}=\ket{\psi_p(\vec{\theta})}$. Assume, for contradiction, that every edge along the path has a zero contribution, i.e.\ $\bra{\psi_p}C_{v_{i-1}v_i}\ket{\psi_p}=0$ for all $i$. Since each $C_{v_{i-1}v_i}$ is a positive semidefinite projector, this implies
\[
C_{v_{i-1}v_i}\ket{\psi_p}=0
\qquad\text{for all } i=1,\ldots,\ell.
\]
Using $C_{uv}=(1-Z_uZ_v)/2$, this is equivalent to
\[
Z_{v_{i-1}}Z_{v_i}\ket{\psi_p}=\ket{\psi_p}
\qquad\text{for all } i=1,\ldots,\ell.
\]
Multiplying these parity constraints along the path, all intermediate $Z_{v_i}^2$ factors cancel, yielding
\[
Z_{v_0}Z_{v_\ell}\ket{\psi_p}=\ket{\psi_p}.
\]
Thus $\bra{\psi_p}Z_{v_0}Z_{v_\ell}\ket{\psi_p}=1$. Since the path is geodesic and $\ell>2p$, this contradicts Lemma~\ref{lem:far_vertex_decorrelation}. Therefore at least one edge on the path must have a positive contribution.
\end{proof}
For bounded-degree graph families and fixed $p$, Lemma~\ref{lem:long_path_positive_edge} already suggests a path-level route to cost extensivity whenever the graph contains linearly many edge-disjoint geodesic paths of length $\ell>2p$.
While Lemma~\ref{lem:long_path_positive_edge} gives positivity \textit{somewhere} along a long geodesic path, Lemma~\ref{lem:no_saturation_edge_type} turns this into a statement about an edge-neighborhood type. Namely, if a radius-$p$ edge-neighborhood type $g_0$ can be realized on every edge of a sufficiently long geodesic path, then this type cannot have zero contribution.
To that end, we use the locality relation in Eq.~\eqref{eq:qaoa_subgraph_cost2} and recall from Sec.~\ref{secsec:qaoa} that the expectation value of an edge $f_{p,jk}$ in a $p$-layer QAOA MaxCut circuit depends solely on its reverse causal cone (RCC), namely the edge-neighborhood $g_p(j,k)$ comprising all nodes and edges within distance $p$ of that edge.
\begin{lemma}[Strict positivity of repeatable edge-neighborhood types along a long geodesic]
\label{lem:no_saturation_edge_type}
Fix $g_0$ to be a radius-$p$ edge-neighborhood type that represents a specific RCC. Suppose there exists a graph $G$ containing a geodesic path $P=(v_0,\ldots,v_\ell)$ with $\ell>2p$ such that every edge on the path shares this local structure $g_0$
\[
g_p(v_{i-1},v_i)=g_0
\qquad 
\forall i=1,\ldots,\ell.
\]
Then, for every fixed parameter vector $\vec{\theta}$,
\[
f_{p,v_{i-1},v_i}(\vec{\theta})=f_{p,g_0}(\vec{\theta})>0 \qquad \forall i=1,\ldots,\ell.
\]
\end{lemma}
\begin{proof}
Suppose, to the contrary, that $f_{p,g_0}(\vec{\theta})=0$. Then the locality statement in Eq.~\eqref{eq:qaoa_subgraph_cost2} implies that every edge on $P$ would yield a zero contribution, contradicting Lemma~\ref{lem:long_path_positive_edge}. 
\end{proof}
Consequently, by locality, \textit{every} edge in \textit{any} graph whose radius-$p$ edge-neighborhood type has the same type $g_0$ as in Lemma~\ref{lem:no_saturation_edge_type} has the same strictly positive contribution. The geodesic path in the lemma is used only as a witness certifying positivity of this local type.
\paragraph{Extensive consequence for locally tree-like families}
Lemma~\ref{lem:no_saturation_edge_type} provides the missing positivity condition used in Reasoning II for graphs satisfying the \BS convergence. If a graph family contains a radius-$p$ edge-neighborhood (RCC) type $g_0$ that can be repeated along a geodesic path of length greater than $2p$ then Lemma~\ref{lem:no_saturation_edge_type} ensures $f_{p,g_0}(\vec{\theta}) > 0$. In this case, Eq.~\eqref{eq:qaoa_sum_over_subgraphs} gives
\[
\bigl\vert \FpThetaGm \bigr \vert
\;=\;
\sum_g w_g(G_m) f_{p,g}(\vec{\theta})
\ge
w_{g_0}(G_m)f_{p,g_0}(\vec{\theta}) \gt 0.
\]
Thus, whenever type $g_0$ appears with positive edge density, leading to $w_{g_0}(G_m)=\Theta(m)$, Assumption~\ref{ass:qaoa_cost_lower_bound} holds with a constant proportional to $f_{p,g_0}(\vec{\theta})$.

\textbf{In the case of $v$-regular locally tree-like graphs}, take $g_0$ to be the radius-$p$ edge-neighborhood in the infinite $v$-regular tree. This type satisfies Lemma~\ref{lem:no_saturation_edge_type}: along any sufficiently long geodesic in the infinite $v$-regular tree, every edge has the same radius-$p$ neighborhood type $g_0$. Therefore $f_{p,g_0}(\vec{\theta})>0$. If a finite $v$-regular graph contains this regular tree RCC $g_0$ on a positive fraction of its edges, then
\[
\bigl\vert \FpThetaG \bigr \vert
\ge
f_{p,g_0}(\vec{\theta})\,w_{g_0}(G).
\]
In particular, for uniformly random $v$-regular graphs and fixed $p$, \BS convergence implies that the fraction of edges with this regular tree structure tends to one with high probability. Hence
\[
\bigl\vert \FpThetaG \bigr\vert
\ge
f_{p,g_0}(\vec{\theta})\,(1-o(1))\,m.
\]
\textbf{In the case of sparse \ER graphs} $G(n,v/n)$ with fixed $v>0$, the edge-rooted local neighborhoods converge in the \BS sense to the corresponding Poisson Galton--Watson edge-rooted tree. Any fixed finite tree RCC that occurs in this local limit has a strictly positive limiting probability $\rho(g_0)>0$. More concretely, let $g_0$ be the radius-$p$ edge-neighborhood of the central edge in the infinite path. In the edge-rooted Poisson Galton--Watson limit, this event has positive probability; with the induced radius-$p$ neighborhood convention used above, $\rho(g_0)=(v e^{-v})^{2p}>0$, since on each side of the root edge the first $p$ forward offspring numbers must equal one. Since $g_0$ can be repeated along any sufficiently long geodesic path, Lemma~\ref{lem:no_saturation_edge_type} ensures $f_{p,g_0}(\vec{\theta})>0$. 
Reasoning II then gives
\[
\mathbb{E}_{G}\!\left[\bigl\vert \FpThetaGm \bigr\vert\right]
\ge
f_{p,g_0}(\vec{\theta})\,\bigl(\rho(g_0)-o(1)\bigr)\,m
=
f_{p,g_0}(\vec{\theta})\,\bigl((v e^{-v})^{2p}-o(1)\bigr)\,m.
\]
Thus both locally tree-like $v$-regular graphs and sparse \ER graph families contain positive-density edge neighborhood types whose contribution is strictly positive by Lemma~\ref{lem:no_saturation_edge_type}, supporting the extensive scaling in Assumption~\ref{ass:qaoa_cost_lower_bound}.

Finally, we conjecture that the same positivity mechanism extends beyond the repeatable-path types covered above. In particular, an edge-neighborhood type should have strictly positive contribution whenever one can move $p$ steps away from the root edge along a simple path and reach an open end.
\begin{conjecture}[Open-ended edge-neighborhood positivity]
\label{conj:open_ended_edge_positive}
Fix $p\ge1$ and a radius-$p$ edge-neighborhood type $g$ with root edge $e=(a,b)$. Suppose that one can leave the root edge and follow a simple path of length $p$,
\[
u_0,u_1,\ldots,u_p,
\]
where $u_0\in\{a,b\}$, $(u_0,u_1)\neq e$, the vertices $u_0,\ldots,u_p$ are all distinct, and $u_p$ has degree one in $g$. Then, for every QAOA parameter vector $\vec{\theta}$,
\[
f_{p,g}(\vec{\theta})>0.
\]
\end{conjecture}
\section{QAOA Concentration bound} \label{appendix:qaoa_concentration_bound}
For completeness, we restate \textit{Lemma} \ref{lem:qaoa_concentration_inequality}.\\
Let $\langle C_p \rangle$ be the expectation of the cost for a $p$-depth QAOA circuit, and let $\langle \hat{C}_p \rangle$ denote its estimator. Let $m$ be the number of edges in the graph, $n_p$ the number of circuit measurements, and $\rvpsym$ the bound from Eq.~\eqref{eq:h_vp_definition}. Then, for any $\delta > 0$, we have:
\begin{equation}
    \Pr\bigl(\Bigm\lvert\langle \hat{C_p} \rangle - \langle C_p \rangle \Bigm\rvert \gt \delta\costexpect\bigr) \leq 2\exp\Bigl(\frac{-2n\delta^2{\langle C_p \rangle}^2}{(2\rvpsym+1) m}\Bigr).
\end{equation}
\begin{proof}
    First, we apply Janson's bound from Eq.~\eqref{eq:janson_bound}:
    \begin{equation}
        \Pr\bigl(\Bigm\lvert X-\mu_X\Bigm\rvert \gt t\bigr) \leq 2\exp\Bigl(\frac{-2t^2}{(\Lambda + 1)\sum\limits_{\alpha \in \mathcal{A}}(b_\alpha-a_\alpha)^2}\Bigr),
    \end{equation}
    This is essentially a Chernoff--Hoeffding bound with a penalty factor $\Lambda + 1$ in the denominator. The parameter $\Lambda$ is not the variance; it is the maximum degree of the dependency graph of the summands. Thus, dependence enters the bound by widening the Hoeffding concentration scale, through the factor $(\Lambda+1)\sum_{\alpha \in \mathcal{A}}(b_\alpha-a_\alpha)^2$.
    \medskip
    \noindent\textbf{Replace dependency graph with QAOA cone causality.}
    Let $\alpha=(i,(j,k))$ index the contribution of edge $(j,k)\in E$ in the $i$th circuit measurement, and define
    $Y_\alpha=C_{jk}^{(i)}=\frac{1}{2}(1-Z_jZ_k)$. Then
    \[
    X=\sum_{i=1}^n\sum_{(j,k)\in E} C_{jk}^{(i)}
    \]
    is the sum of all measured edge contributions over the $n$ shots. Since every edge contribution is bounded as $0\le C_{jk}^{(i)}\le 1$, Janson's bound gives
    \begin{equation}
        \Pr\bigl(\Bigm\lvert X-\mu_X\Bigm\rvert \gt t\bigr) \leq 2\exp\left(\frac{-2t^2}{(\Lambda+1)\sum\limits_{i=1}^n\sum\limits_{(j,k)\in E}(1-0)^2}\right).
    \end{equation}
    \\\\
    In QAOA MaxCut, two edge-cost measurement outcomes $C_{jk}$ and $C_{j^\prime k^\prime}$ are independent whenever their corresponding depth-$p$ causal cones, denoted by $g(j,k)$ and $g(j^\prime,k^\prime)$, are vertex-disjoint (see Eq.~\eqref{eq:qaoa_statistical_independence}). Therefore, the dependency graph for the edge measurements has maximum degree bounded by the number of edge terms whose causal cones can overlap a given edge's causal cone. Using the same locality counting as in Eq.~\eqref{eq:variance_upper_bound} from \cite{farhi2014quantum}, this gives $\Lambda \le 2\rvpsym$, and hence
    \[
        \Pr\bigl(\Bigm\lvert X-\mu_X\Bigm\rvert \gt t\bigr) \leq 2\exp\Bigl(\frac{-2t^2}{(2\rvpsym+1) n m}\Bigr),
    \]
    The cost estimator is then defined as $\hat{x}_n=\frac{X}{n}$, and its expectation is $\mu = \frac{\mu_X}{n}$. The event $\lvert \hat{x}_n-\mu\rvert>\delta\,\abs{\mu}$ is equivalent to $\lvert X-\mu_X\rvert>n\delta\,\abs{\mu}$, so we set $t=n\delta\,\abs{\mu}$ and obtain
    \[
        \Pr\bigl(\Bigm\lvert X-\mu_X\Bigm\rvert \gt n\delta\,\abs{\mu}\bigr) \leq 2\exp\Bigl(\frac{-2n\delta^2 \mu^2}{(2\rvpsym+1) m}\Bigr),
    \]
    which is the same as 
    \[
        \Pr\bigl(\Bigm\lvert\hat{x}_n-\mu \Bigm\rvert \gt \delta\,\abs{\mu}\bigr) \leq 2\,exp\Bigl(\frac{-2n\delta^2 \mu^2}{(2\rvpsym+1) m}\Bigr). 
    \]
    We then rewrite in terms of cost estimator and expectation
    \[
    \Pr\bigl(\Bigm\lvert\langle \hat{C_p} \rangle - \langle C_p \rangle \Bigm\rvert \gt \delta\costexpect\bigr) \leq 2\exp\Bigl(\frac{-2n\delta^2{\langle C_p \rangle}^2}{(2\rvpsym+1) m}\Bigr).
    \]
\end{proof}
\section{QAOA convergence in a stochastic gradient descent process}
\subsection{Proving Lemma \ref{lem:relative_sgd_convergence}}
\label{app:lem_relative_sgd_convergence_proof}
We measure relative performance w.r.t.\ $\abs{F(\vec{\theta}^*)}$ to avoid sign issues since in our convention $F$ is typically non-positive.
For completeness, we restate Lemma~\ref{lem:relative_sgd_convergence}:
\medskip
Let $F(\vec{\theta})$ be a cost function with an $L$-Lipschitz gradient that satisfies the PL-inequality for some $\mu > 0$ (see Eq.~\eqref{eq:PL_inequality}). Consider SGD with learning rate $\alpha = 1/L$. Then, any iteration $t$ satisfies
\begin{equation*}
    d_t \,\bigl\vert F(\vec{\theta}^*)\bigr\vert
    \;\leq\;
    d_0\left(1-\frac\mu L\right)^t
    \bigl\vert F(\vec{\theta}^*)\bigr\vert
    \;+\;
    \xi \,\bigl\vert F(\vec{\theta}^*)\bigr\vert,
\end{equation*}
where $d_t$ is the relative cost gap at iteration $t$, as defined in Eq.~\eqref{eq:d_t_relative_cost_gap}, and $\xi$ is the \emph{relative stochastic error floor}, defined as
\begin{equation*}
   \xi \;\overset{\text{def}}{=}\; \frac{\sigma^2}{2\mu \,\bigl\vert F(\vec{\theta}^*)\bigr\vert},
\end{equation*}
where $\sigma^2$ is an upper bound on the mean-squared error of the gradient estimate
\begin{equation*}
    \mathbb{E}\!\left[\bigl\|\hat g^{\,t}-\nabla F(\vec{\theta}^{\,t})\bigr\|^2\right] \leq \sigma^2.
\end{equation*}

\medskip
\begin{proof}
We use $\mathbb{E}_t[\cdot] := \mathbb{E}[\cdot \mid \vec{\theta}^{\,t}]$ for conditional expectation over the probability distribution of $\vec{\theta}$ at time $t$.
Define the gradient-estimation error
\[
\varepsilon^{\,t} \;:=\; \hat g^{\,t} - \nabla F(\vec{\theta}^{\,t}).
\]
Notice that $\varepsilon^{\,t}$ may be \emph{biased}, i.e.\ $\mathbb{E}_t[\varepsilon^{\,t}] \neq 0$. For finite-difference gradients this bias corresponds to truncation error (scaling as $\mathcal{O}(h^k)$), and its contribution is included in the MSE via $\mathbb{E}\|\varepsilon^{\,t}\|^2$.
In the main text we choose $h$ so that this bias contribution is negligible, but the bound below does not require unbiasedness.
\smallskip

\noindent\textbf{One-step progress via $L$ smoothness.}
In optimization theory, the \textit{descent lemma}, see e.g.\   \cite{nesterov2004introductory},  states that for $L$-Lipschitz continuous gradient function $F$, as defined in Eq.~\eqref{eq:lipschitz_continous}, 
the following holds for any   $\vec{x},\vec{y} \in \mathbb{R}^n$
\[
F(\vec{y}) \le F(\vec{x}) + \langle \nabla F(\vec{x}), \vec{y}-\vec{x}\rangle + \frac{L}{2}\|\vec{y}-\vec{x}\|^2.
\]
Assign $\vec{x}=\vec{\theta}^{\,t}$ and $\vec{y}=\vec{\theta}^{\,t+1}=\vec{\theta}^{\,t}-\alpha \hat g^{\,t}$:
\[
F(\vec{\theta}^{\,t+1})
\le
F(\vec{\theta}^{\,t})
- \alpha \langle \nabla F(\vec{\theta}^{\,t}), \hat g^{\,t}\rangle
+ \frac{\alpha^2 L}{2} \|\hat g^{\,t}\|^2.
\]
Taking $\mathbb{E}_t[\cdot]$ and substituting $\hat g^{\,t}=\nabla F(\vec{\theta}^{\,t})+\varepsilon^{\,t}$ gives
\begin{align*}
\mathbb{E}_t\!\left[F(\vec{\theta}^{\,t+1})\right]
&\le
F(\vec{\theta}^{\,t})
-\alpha\left\langle \nabla F(\vec{\theta}^{\,t}),
\nabla F(\vec{\theta}^{\,t})+\mathbb{E}_t[\varepsilon^{\,t}]\right\rangle \\
&\quad
+ \frac{\alpha^2 L}{2}\,
\mathbb{E}_t\!\left[\bigl\|\nabla F(\vec{\theta}^{\,t})+\varepsilon^{\,t}\bigr\|^2\right] \\
&=
F(\vec{\theta}^{\,t})
+\Bigl(-\alpha+\frac{\alpha^2 L}{2}\Bigr)\|\nabla F(\vec{\theta}^{\,t})\|^2 \\
&\quad
+\Bigl(-\alpha+\alpha^2 L\Bigr)\left\langle \nabla F(\vec{\theta}^{\,t}), \mathbb{E}_t[\varepsilon^{\,t}]\right\rangle
+\frac{\alpha^2 L}{2}\mathbb{E}_t\!\left[\|\varepsilon^{\,t}\|^2\right].
\end{align*}
Now set $\alpha=1/L$. Then the bias-dependent cross term cancels exactly since $-\alpha+\alpha^2 L=0$, yielding
\[
\mathbb{E}_t\!\left[F(\vec{\theta}^{\,t+1})\right]
\le
F(\vec{\theta}^{\,t})
-\frac{1}{2L}\|\nabla F(\vec{\theta}^{\,t})\|^2
+\frac{1}{2L}\mathbb{E}_t\!\left[\|\varepsilon^{\,t}\|^2\right].
\]
Using $\mathbb{E}\!\left[\|\varepsilon^{\,t}\|^2\right]=\mathbb{E}\!\left[\|\hat g^{\,t}-\nabla F(\vec{\theta}^{\,t})\|^2\right]\le\sigma^2$ gives
\[
\mathbb{E}\!\left[F(\vec{\theta}^{\,t+1})\right]
\le
\mathbb{E}\!\left[F(\vec{\theta}^{\,t})\right]
-\frac{1}{2L}\mathbb{E}\!\left[\|\nabla F(\vec{\theta}^{\,t})\|^2\right]
+\frac{\sigma^2}{2L}.
\]
\smallskip
\noindent\textbf{From gradients to function values via PL.}
Subtract $F(\vec{\theta}^*)$ and define $\Delta_t:=\mathbb{E}[F(\vec{\theta}^{\,t})]-F(\vec{\theta}^*)$.
Since the PL inequality holds pointwise,
$\|\nabla F(\vec{\theta}^{\,t})\|^2 \ge 2\mu\bigl(F(\vec{\theta}^{\,t})-F(\vec{\theta}^*)\bigr)$,
taking expectation yields
\[
\mathbb{E}\!\left[\|\nabla F(\vec{\theta}^{\,t})\|^2\right]\ge 2\mu\,\bigl(\mathbb{E}[F(\vec{\theta}^{\,t})]-F(\vec{\theta}^*)\bigr).
\]
Therefore,
\[
\mathbb{E}[F(\vec{\theta}^{\,t+1})]-F(\vec{\theta}^*)
\le
\Bigl(1-\frac{\mu}{L}\Bigr)\bigl(\mathbb{E}[F(\vec{\theta}^{\,t})]-F(\vec{\theta}^*)\bigr)
+\frac{\sigma^2}{2L}.
\]
\smallskip
\noindent\textbf{Unrolling the recursion.}
Let $r:=1-\mu/L\in[0,1)$. Iterating gives
\[
\Delta_t \le r^t \Delta_0 + \frac{\sigma^2}{2L}\sum_{k=0}^{t-1} r^k
\le r^t \Delta_0 + \frac{\sigma^2}{2L}\cdot \frac{1}{1-r}
=
r^t \Delta_0 + \frac{\sigma^2}{2\mu}.
\]
Dividing by $\abs{F(\vec{\theta}^*)}$ and using
$d_t=\Delta_t/\abs{F(\vec{\theta}^*)}$ and $\xi=\sigma^2/(2\mu\abs{F(\vec{\theta}^*)})$ yields
\[
    d_t \,\bigl\vert F(\vec{\theta}^*)\bigr\vert
    \;\leq\;
    d_0\left(1-\frac\mu L\right)^t
    \bigl\vert F(\vec{\theta}^*)\bigr\vert
    \;+\;
    \xi \,\bigl\vert F(\vec{\theta}^*)\bigr\vert,
\]
which is equivalent to the stated inequality.
\end{proof}

    \subsection{Linear scaling of the QAOA PL $\mu$ and the Lipschitz constant $L$}
    \textbf{Restating Assumption~\ref{ass:linear_scaling_mu}:}
    \newline
    \label{app:linear_scaling_mu_L}
    The expected QAOA cost at depth $p$, $F_p(\vec{\theta})$, satisfies the PL inequality (Eq.~\eqref{eq:PL_inequality}) with a constant $\mu > 0$ that scales linearly with the size of the graph:
    \begin{equation}
        \mu = \Theta(m).
    \end{equation}
    
    \medskip\noindent
    \textbf{Restating Assumption~\ref{ass:linear_scaling_L}:}
    \newline
    The expected QAOA cost at depth $p$, $F_p(\vec{\theta})$, satisfies the PL inequality (Eq.~\eqref{eq:PL_inequality}) and has an $L$-Lipschitz continuous gradient which scales linearly with the size of the graph:
    \begin{equation}
        L = \Theta(m).
    \end{equation}
    \paragraph{Reasoning: Benjamini-Schramm local limit.}
    Let $(G_m)_{m\ge 1}$ be a sequence of graphs with a fixed maximal degree, indexed by the number of edges $m:=\lvert E(G_m) \rvert$,
    and let $e_m$ be an edge uniformly sampled from $G_m$. We assume $(G_m)$ is Benjamini-Schramm
    convergent (in the edge-rooted sense, as described in Appendix~\ref{app:cost_function_scaling}, in \textit{Reasoning II}) for a fixed radius-$p$ edge-neighborhood. We utilize the depth-$p$ expected QAOA cost in the local-decomposition form of Eq.~\eqref{eq:qaoa_expectaion_edges_sum}
    \[
    F_p(\vec{\theta})\;=\;  F_p(\vec{\theta};G_m) \;=\; -\sum_{g} w_g\, f_{p,g}(\vec{\theta}),
    \]
    where $w_g=w_g(G_m)$ counts the number of radius-$p$ edge-neighborhoods of type $g$, and $f_{p,g}(\vec{\theta})$ is the corresponding local contribution. 
    
    Using the property of \BS convergence, as the graph grows, the fraction of each neighborhood type $g$ converges:
    \[
    \rho_m(g):=\frac{w_g}{m}\;
    \xrightarrow[m \to \infty]{}
    \rho(g),
    \]
    where $\rho(g)$ is the limiting probability of seeing type $g$ around a uniformly sampled edge. 
    Consequently, we define the average local cost per edge as:
    \begin{equation}
    \label{eq:averaged_per_edge_cost}
     \bar f_{p}(\vec{\theta};G_m):=\frac{1}{m}F_p(\vec{\theta};G_m)
    = -\sum_g \frac{w_g}{m}\, f_{p,g}(\vec{\theta})
    =-\sum_g \rho_m(g)\, f_{p,g}(\vec{\theta}),   
    \end{equation}
    as well as the limiting per-edge cost function based on the limiting probability:
    \[
    \bar f_{p}(\vec{\theta}):=-\sum_g \rho(g)\, f_{p,g}(\vec{\theta}).  
    \]
    Since $p$ and the maximum degree are fixed, the sum involves only finitely many neighborhood types. Moreover, since each local contribution $f_{p,g}(\vec{\theta})$ is a smooth trigonometric polynomial of the QAOA parameters, with bounded derivatives on the chosen parameter range, 
    the convergence of the densities $\rho_m(g)$ directly ensures that the per-edge cost, its gradient, and its Hessian converge to their respective limits uniformly over the chosen parameter range:
    \begin{align}
    \bar f_{p}(\vec{\theta};G_m) &\xrightarrow[m \to \infty]{}\bar f_{p}(\vec{\theta}), \nonumber \\
    \nabla \bar f_{p}(\vec{\theta};G_m) &\xrightarrow[m \to \infty]{}\nabla \bar f_{p}(\vec{\theta}), \\
    \nabla^2 \bar f_{p}(\vec{\theta};G_m) &\xrightarrow[m \to \infty]{}\nabla^2 \bar f_{p} \nonumber (\vec{\theta}).
    \end{align}
    \smallskip
    which, in terms of scaling with $m$, implies:
    \begin{equation}
    \label{eq:cost_and_grad_scaling}
    \begin{aligned}
    \bar f_p(\vec{\theta};G_m) &= \bar f_p(\vec{\theta})+o(1), \\
    \nabla \bar f_p(\vec{\theta};G_m) &= \nabla \bar f_p(\vec{\theta})+o(1), \\
    \nabla^2 \bar f_p(\vec{\theta};G_m) &= \nabla^2 \bar f_p(\vec{\theta})+o(1).
    \end{aligned}
    \end{equation}
    We now assume that the limiting per-edge objective $\bar f_p(\vec{\theta})$ is strongly convex (no flat regions) in a local optimization region $U$ around a local optimum $\vec{\theta}^*$, namely
    \[
    \nabla^2\bar f_p(\vec{\theta})\succeq \lambda_0 I,
    \qquad \vec{\theta}\in U.
    \]
    Here $\lambda_0>0$ is a lower bound on the Hessian eigenvalues in $U$, and is independent of the graph size $m$.
    By the Hessian convergence in Eq.~\eqref{eq:cost_and_grad_scaling}, for all sufficiently large $m$,
    \[
    \nabla^2\bar f_p(\vec{\theta};G_m)\succeq \frac{\lambda_0}{2}I,
    \qquad \vec{\theta}\in U.
    \]
    Therefore, $\bar f_p(\vec{\theta};G_m)$ is also strongly convex on $U$, and hence satisfies the PL inequality on $U$ with a constant $\bar\mu_m\ge\lambda_0/2=\Omega(1)$ \cite{wolf_mathematical}:
    \[
    \frac{1}{2}\|\nabla \bar f_{p}(\vec{\theta};G_m)\|_2^2
    \ge
    \bar\mu_m\big(\bar f_{p}(\vec{\theta};G_m)-\bar f_{p}(\vec{\theta}_m^\ast;G_m)\big),
    \]
    where $\vec{\theta}_m^\ast\in\arg\min_{\vec{\theta}\in U}\bar f_p(\vec{\theta};G_m)$, equivalently $\vec{\theta}_m^\ast\in\arg\min_{\vec{\theta}\in U}F_p(\vec{\theta};G_m)$. Since
    \[
    F_p(\vec{\theta};G_m)=m\,\bar f_p(\vec{\theta};G_m),
    \qquad
    \nabla F_p(\vec{\theta};G_m)=m\,\nabla \bar f_p(\vec{\theta};G_m),
    \]
    we obtain
    \begin{eqnarray*}
    \frac{1}{2}\|\nabla F_p(\vec{\theta};G_m)\|_2^2
    &=&
    \frac{m^2}{2}\|\nabla \bar f_p(\vec{\theta};G_m)\|_2^2
    \\
    &\ge&
    m^2\bar\mu_m
    \bigl(\bar f_p(\vec{\theta};G_m)-\bar f_p(\vec{\theta}_m^\ast;G_m)\bigr)
    \\
    &=&
    m\bar\mu_m
    \bigl(F_p(\vec{\theta};G_m)-F_p(\vec{\theta}_m^\ast;G_m)\bigr).
    \end{eqnarray*}
    
    Thus, the total objective $F_p(\vec{\theta};G_m)$ satisfies the PL inequality on the same local region $U$ with
    \[
    \mu=m\,\bar\mu_m=\Omega(m), 
    \]
        which gives the required lower bound on the PL constant. 
    
    The matching upper bound, namely $\mu=\mathcal{O}(m)$, will follow from the Lipschitz argument below, since any PL constant of an $L$-smooth function satisfies $\mu\le L$.
    To that end, we now 
    prove the upper bound $L=\mathcal{O}(m)$ for the Lipschitz constant $L$ of the total cost function $F_p$. 
    The matching lower bound, i.e.\ $L=\Omega(m),$ required to show Assumption~\ref{ass:linear_scaling_L}, then follows directly from the PL scaling established above and the same inequality $\mu\le L$.
    Starting from the local cost contributions $f_{p,g}(\vec{\theta})$, note that for fixed depth $p$ and bounded maximum degree, each $f_{p,g}(\vec{\theta})$ is a smooth function of the QAOA parameters $\vec \theta$, so its gradient-Lipschitz constant $L_g$ is bounded by some constant $L^*$:   
    \begin{equation}
    \label{eq:local_cost_lipschitz_bound}
    \bigl\|\nabla f_{p,g}(\vtheta_1)-\nabla f_{p,g}(\vtheta_2)\bigr\| 
    \le L_g\|\vtheta_1-\vtheta_2\| \leq
    L^*\|\vtheta_1-\vtheta_2\| \quad \forall g,\vtheta_1,\vtheta_2.
    \end{equation} 
    Since $L^*$ is determined solely by the finite set of radius-$p$ neighborhood types allowed by the fixed depth and bounded maximum degree, it is invariant of graph size, i.e., $L^*=\mathcal{O}(1)$.
    Next, we denote by $\bar L_m$ the gradient-Lipschitz constant of the averaged per-edge cost $\bar f_p(\vtheta;G_m)$. 
    From the linearity of the gradient and Eq.~\eqref{eq:averaged_per_edge_cost} we have 
    \[
    \nabla \bar f_p(\vec{\theta};G_m)
    =
    -\sum_g \frac{w_g(G_m)}{m}\,\nabla f_{p,g}(\vec{\theta}), 
    \]
    which, using the triangle inequality and Eq.~\eqref{eq:local_cost_lipschitz_bound}, gives 
    \[
    \begin{aligned}
    \bigl\|\nabla \bar f_p(\vec{\theta}_1;G_m)&-\nabla \bar f_p(\vec{\theta}_2;G_m)\bigr\| \\
    &\le
    \sum_g \frac{w_g(G_m)}{m}
    \bigl\|\nabla f_{p,g}(\vec{\theta}_1)-\nabla f_{p,g}(\vec{\theta}_2)\bigr\| \\
    & \le \underbrace{\sum_g \frac{w_g(G_m)}{m}}_{1}
     L^*  \bigl\|\vec{\theta}_1-\vec{\theta}_2\bigr\|
    =
    L^*\bigl\|\vec{\theta}_1-\vec{\theta}_2\bigr\|, \qquad \forall  \vec{\theta}_1, \vec{\theta}_2
    \end{aligned}
    \]
    where we used the fact that $\sum_g w_g(G_m)=m$. 
    By definition, the gradient-Lipschitz constant of the averaged per-edge cost, $\bar L_m$, is the smallest value satisfying:
    \[\left\| \nabla \bar f_p(\vec{\theta}_1;G_m) - \nabla \bar f_p(\vec{\theta}_2; G_m) \right\| \le \bar L_m \| \vec{\theta}_1 - \vec{\theta}_2 \|,
    \]
     from which it follows that $\bar L_m \leq L^*$.  Combined with the previous observation that $L^* = \mathcal{O}(1)$, this ensures that $\bar L_m$ is independent of graph size, that is $\bar L_m = \mathcal{O}(1)$.
    Finally, we consider the gradient-Lipschitz constant $L$ of the total cost function $F_p(\vec{\theta};G_m)$. Since $F_p(\vec{\theta};G_m)=m\,\bar f_p(\vec{\theta};G_m)$, 
    we have
    \[
    \begin{aligned}
    \bigl\|\nabla F_p(\vec{\theta}_1;G_m)&-\nabla F_p(\vec{\theta}_2;G_m)\bigr\| \\
    &=
    m\bigl\|\nabla \bar f_p(\vec{\theta}_1;G_m)-\nabla \bar f_p(\vec{\theta}_2;G_m)\bigr\| \\
    &\le
    m\bar L_m\|\vec{\theta}_1-\vec{\theta}_2\|.
    \end{aligned}
    \]
    Hence $F_p(\vec{\theta};G_m)$ is $L$-smooth with $L \leq m\bar L_m = \mathcal{O}(m)$, which proves the upper bound. 
    Moreover, any function that is both $L$-smooth and satisfies the PL inequality with constant $\mu$ obeys $\mu\le L$. Since the PL argument above gives $\mu=\Theta(m)$, it follows that $L=\Omega(m)$, which finally gives
    \[
    L=\Theta(m), 
    \]
    consistent with Assumption~\ref{ass:linear_scaling_L}.
    Moreover, together with $\mu\le L=\mathcal{O}(m)$, the lower bound $\mu=\Omega(m)$ also gives 
    \[\mu=\Theta(m),
    \] as stated in Assumption~\ref{ass:linear_scaling_mu}.
\subsection{Results of SGD with the \textit{Parameter Shift Rule}}
\label{app:parameter_shift_rule_results}
For completeness, we include the analogue of Fig.~\ref{fig:3regular_sgd_finite_difference} for the parameter-shift estimator. The goal here is to compare the mean values and spreads of the $d_t$ trajectories between the two gradient-estimation methods, namely the parameter-shift rule (Fig.~\ref{fig:2regular_sgd_param_shift}) and the finite-difference approximation (Fig.~\ref{fig:2regular_sgd_finite_diff_appendix}), under the same inverse shot schedule. Unlike the main-text 3-regular experiment, we use here cycle graphs because, in our simulations, the parameter-shift estimator did not reliably converge for the larger random 3-regular instances.
From Eq.~\eqref{eq:shots_per_case} it is expected that maintaining a fixed target relative cost gap in the parameter-shift setting requires an approximately size-independent shot budget, while for the finite-difference approximation, an inverse shot schedule is sufficient. To empirically test this prediction, the sampled cost evaluations entering the gradient estimator are intentionally assigned a shot budget $\shots$ that scales inversely with the number of edges, i.e.\ $\shots \propto 1/m$ (as in Fig.~\ref{fig:3regular_sgd_finite_difference}). This allows us to explicitly visualize the resulting growth of the trajectory spread in the parameter-shift rule setting, while observing that it remains stable in the finite-difference case.
Overall, the distinct behaviors displayed across Figs.~\ref{fig:2regular_sgd_param_shift} and~\ref{fig:2regular_sgd_finite_diff_appendix} under the same inverse-shot schedule are consistent with the different shot-scaling  prescriptions of Eq.~\eqref{eq:shots_per_case}.
   \begin{figure}[H]
        \centering
        \subfloat[\textbf{Parameter-shift rule:} Number of shots scales inversely with the number of edges in \textbf{cycle graphs}.]
        {\includegraphics[width=0.45\linewidth]{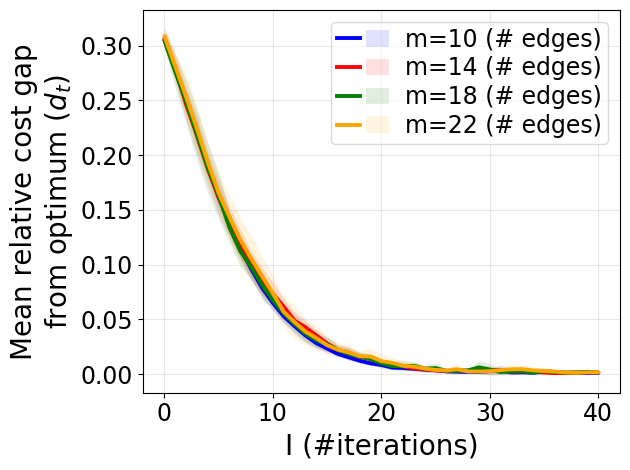}}
        \quad
       \subfloat[Higher resolution: the corresponding standard deviation $\sigma(d_t)$ as a function of iterations.]
       {\includegraphics[width=0.45\linewidth]{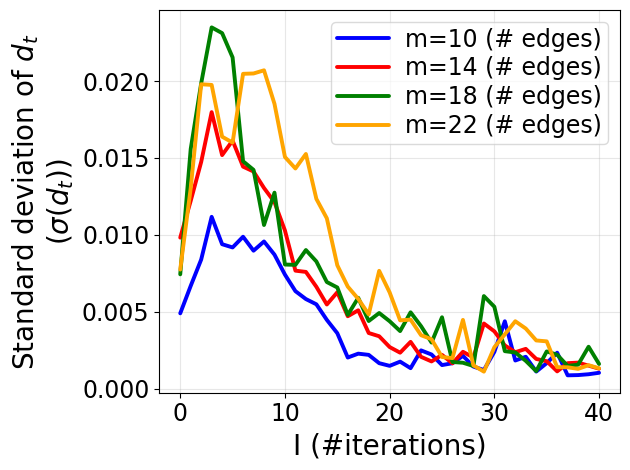}}
       \caption[Representative SGD trajectories with parameter-shift estimation]
        {\protect\textbf{Relative cost gap ($d_t$) trajectories on cycle graphs using parameter-shift estimation under inverse shot scaling.} Cycle graphs with $N\in\{10,14,18,22\}$ nodes, corresponding to $m=N$ edges, were considered, with $10$ independent finite-shot SGD trajectories per graph size, at circuit depth $p=1$. All trajectories were initialized from the same fixed parameter vector $(\gamma_1,\beta_1)=(0,0.5)$. The values $d_t$ were evaluated from exact state-vector cost expectations along the SGD trajectories, using the corresponding noiseless reference value according to Eq.~\eqref{eq:d_t_relative_cost_gap}. The parameter-shift gradient estimator was employed, using a number of shots per circuit evaluation proportional to $1/m$. \textbf{(a)} The mean relative cost gap $d_t$ is shown as a function of the iteration number (solid lines) together with the standard deviation $\pm\,\sigma(d_t)$ (shaded). \textbf{(b)} The corresponding standard deviation $\sigma(d_t)$. Across the tested sizes, the convergence rates remain broadly comparable, consistent with the approximately size-independent iteration scaling, while the trajectory spread increases with graph size, illustrating why Eq.~\eqref{eq:shots_per_case} prescribes an approximately size-independent shot budget in the parameter-shift setting.}    
       \label{fig:2regular_sgd_param_shift}
    \end{figure}
    \begin{figure}[H]
        \centering
        \subfloat[\textbf{Finite-difference approximation:} Number of shots scales inversely with the number of edges in \textbf{cycle graphs}.]
        {\includegraphics[width=0.45\linewidth]{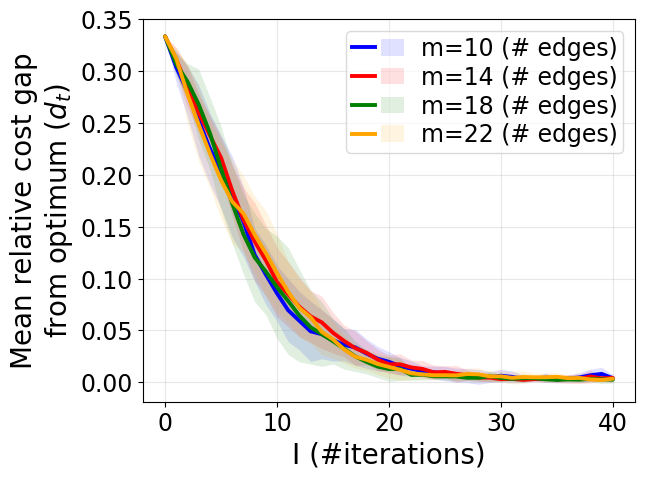}}
        \quad
       \subfloat[Higher resolution: the corresponding standard deviation $\sigma(d_t)$ as a function of iterations.]
       {\includegraphics[width=0.45\linewidth]{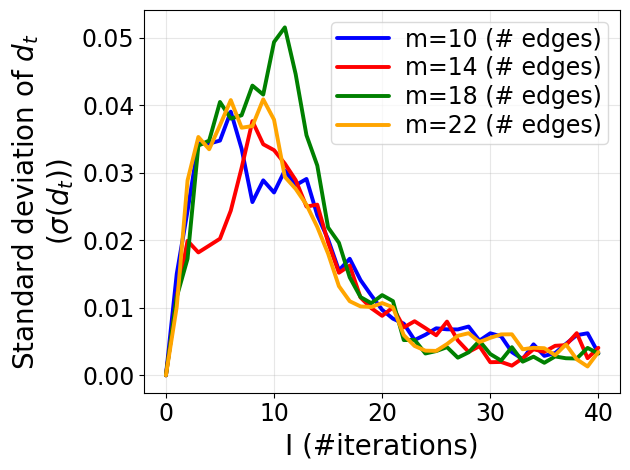}}
       \caption[Representative SGD trajectories using finite-difference estimation]{\protect\textbf{Relative cost gap ($d_t$) trajectories on cycle graphs using finite-difference estimation under inverse shot scaling.} Cycle graphs with $N\in\{10,14,18,22\}$ nodes, corresponding to $m=N$ edges, were considered, with 10 independent finite-shot SGD trajectories per graph size, at circuit depth $p=1$. All trajectories were initialized from the same fixed parameter vector $(\gamma_1,\beta_1)=(0,0.5)$. The values $d_t$ were evaluated from exact state-vector cost expectations along the SGD trajectories, using the corresponding noiseless reference value according to Eq.~\eqref{eq:d_t_relative_cost_gap}. A finite-difference gradient estimator was employed, using a number of shots per circuit evaluation proportional to $1/m$, as in Fig.~\ref{fig:2regular_sgd_param_shift}. \textbf{(a)} The mean relative cost gap $d_t$ is shown as a function of the iteration number (solid lines) together with the standard deviation $\pm\,\sigma(d_t)$ (shaded). \textbf{(b)} The corresponding standard deviation $\sigma(d_t)$. Compared with the parameter-shift estimator in Fig.~\ref{fig:2regular_sgd_param_shift}, the trajectory spread remains stable across graph sizes, consistent with Eq.~\eqref{eq:shots_per_case}, which predicts that inverse shot scaling is sufficient for the finite-difference setting.}
       \label{fig:2regular_sgd_finite_diff_appendix}
    \end{figure}
\paragraph{Using the general parameter shift rule}
In the main text, we used the parameter-shift rule, in its original, gate-based formulation. Here, we analyze the \textit{general} parameter-shift rule \cite{Wierichs_2022}. 
Consider a unitary evolution of the form $U(\theta)=e^{-i\theta A}$, where $A$ is a Hermitian operator. The expectation value $\langle \mathcal{H}\rangle_{\psi(\theta)}$ admits a finite Fourier series expansion with frequencies determined by the pairwise differences between the eigenvalues of $A$. 
When the induced positive frequency set is equally spaced, we write it as $\{\Delta,2\Delta,\ldots,R\Delta\}$, where $R$ is the number of positive frequencies and $\Delta$ is the spacing between adjacent frequencies.
We consider only this equidistant-shifts case in the following analysis, a more general treatment can be found in \cite{Wierichs_2022}. The exact derivative can then be written as:
\begin{equation}
    \label{eq:gen_psr}
    \frac{\partial \langle \mathcal{H}\rangle_{\psi(\theta)}}{\partial \theta}
    =
    \sum_{j=1}^{2R}
        c_j\,
        \langle \mathcal{H}\rangle_{\psi(\theta + s_j)},
    \quad
    s_j := \frac{2j-1}{2R}\frac{\pi}{\Delta},\quad
    c_j := \Delta\frac{(-1)^{j-1}}{4R\sin^2\!\left(\frac{2j-1}{4R}\pi\right)},
\end{equation}
where the known normalized formula with shifts $s_j=\frac{(2j-1)\pi}{2R}$ is a special case of $\Delta=1$.
For QAOA MaxCut, the frequency count $R$ scales at least linearly with the graph size. Specifically, the frequency spectrum of the mixer Hamiltonian $B$ scales linearly with the number of nodes; its eigenvalues $\{-N,-N+2,\dots,N\}$ yield a set of positive pairwise differences $\{2,4,\dots,2N\}$, resulting in $R=N$. Additionally, for the problem Hamiltonian $C$, the number of relevant positive frequencies scales at most linearly with the number of edges $m$~\cite{Wierichs_2022}. 
Using Eq.~\eqref{eq:gen_psr} and following Eqs.~\eqref{eq:partial_derivative_variance} and \eqref{eq:ws_scaling_cases}, we define $w_s^{\text{gps}}$, the prefactor of the variance of the partial derivative estimator for the general parameter-shift case and lower bound it as follows: 
\[
w_s^{\mathrm{gps}}
=
\sum_{j=1}^{2R}c_j^2
\ge {c_1}^2 = \frac{\Delta^2}{16R^2\sin^4\!\left(\frac{2-1}{4R}\pi\right)}
\approx \frac{\Delta^2}{16R^2\!\left(\frac{\pi}{4R}\right)^4} = \frac{16R^2\Delta^2}{\!\pi^4}, 
\]
using $sin(x)\approx x$ for $x<<1$.
This means $w_s^{\mathrm{gps}}=\Omega(\Delta^2R^2)$.
For QAOA MaxCut on bounded-degree graph families with $N=\Theta(m)$, this gives $w_s^{\mathrm{gps}}=\Omega(m^2)$. 
\end{appendices}
\bibliography{sn-bibliography}
        
\newpage

\end{document}